\documentclass[12pt]{article}
\usepackage[latin1]{inputenc}
\usepackage{graphicx}
\usepackage{color}
\usepackage{geometry}  
\usepackage{amsmath, amssymb} 
\usepackage{setspace} 
\usepackage{booktabs} 
\usepackage[round]{natbib} 
\usepackage{authblk}

\newcommand{\bs}[1]{\ensuremath{\boldsymbol{#1}}} 
\newcommand*{\QEDA}{\hfill\ensuremath{\square}}

\usepackage{amsthm}
\newtheorem{proposition}{Proposition}
\newtheorem{defin}{Definition}
\newtheorem{cor}{Corollary}
\usepackage{floatrow}         
\usepackage{subfig}
\theoremstyle{definition}
\newtheorem{remark}{Remark}

\usepackage{subfig}
\usepackage{booktabs}

\begin{document}

\title{Modeling the Ratio of Correlated Biomarkers Using Copula Regression}
\author[1]{Moritz Berger}
\author[2]{Nadja Klein}
\author[3]{Michael Wagner}
\author[1,3]{Matthias Schmid}
\affil[1]{Department of Medical Biometry, Informatics and Epidemiology, Medical Faculty, University of Bonn, Venusberg Campus 1, 53127 Bonn}
\affil[2]{Chair of Uncertainty Quantification and Statistical Learning, Research Center for Trustworthy Data Science and Security (UA-Ruhr) and Department of Statistics (Technische Universit\"at Dortmund), Joseph-von-Fraunhofer Stra{\ss}e 25, 44227 Dortmund}
\affil[3]{German Center for Neurodegenerative Diseases, Venusberg Campus 1, 53127 Bonn}
\renewcommand\Affilfont{\itshape\small}

\date{}


\maketitle

\begin{abstract}
\noindent Modeling the ratio of two dependent components as a function of covariates is a frequently pursued objective in observational research. Despite the high relevance of this topic in medical studies, where biomarker ratios are often used as surrogate endpoints for specific diseases, existing models are based on oversimplified assumptions, assuming e.g.\@ independence or strictly positive associations between the components. In this paper, we close this gap in the literature and propose a regression model where the marginal distributions of the two components are linked by Frank copula. A key feature of our model is that it allows for both positive and negative correlations between the components, with one of the model parameters being directly interpretable in terms of Kendall's rank correlation coefficient. We study our method theoretically, evaluate finite sample properties in a simulation study and demonstrate its efficacy in an application to diagnosis of Alzheimer's disease via ratios of amyloid-beta and total tau protein biomarkers.
\end{abstract}

\section{Introduction}

A common objective in observational research is to analyze the ratio of two (possibly dependent) components $U,V \in \mathbb{R}^+$ \citep{qilong}. Typical examples are, among others, (i) the LDL/HDL cholesterol ratio in cardiovascular research \citep{nata03}, defined as the ratio of the low-density lipoprotein and the high-density lipoprotein concentrations in plasma or serum, (ii) the CD4/CD8 ratio in HIV research \citep{caby2016}, which measures the ratio of T helper cells to cytotoxic T cells in the human immune system, (iii) the testosterone over epitestosterone (T/E) ratio in antidoping research \citep{sottas}, and (iv) the GEFC/REFC ratio in ophthalmic research, corresponding to the green and red emission components in fundus autofluorescence imaging \citep{wintergerst2022}. In many of such studies, biomarker ratios are used as early indicators or even as surrogate endpoints for a specific disease. In these cases, the focus is not only on the characterization of the marginal ratio distribution, but also on modeling this distribution as a function of a set of covariates $X=(X_1,\hdots,X_p)^\top$. Usually, this amounts to specifying a regression model that includes the ratio as outcome variable \citep{berger2019}.

When setting up a model relating the ratio outcome $R= U/V$ to the covariates~$X$, a common assumption is that both components follow either log-normal or gamma distributions, thereby accounting for the nonnegativity of the component values and the skewness of their distributions \citep{mitchell, vandomelen}. In the former case it is easily derived that the ratio is again log-normally distributed. The latter case, which will be dealt with in this paper, is considerably less straightforward but is often preferred in practice due to its increased efficiency \citep{firth, wiens, berger2019}.

In the special case where $U$ and $V$ are {\em independently} gamma distributed, the ratio $R=U/V$ follows a generalized beta distribution of the second kind, in the following abbreviated by \emph{GB2} \citep{KleiberKotz}. A regression approach for the GB2 distribution has been proposed by \citet{Tulu2013}, who studied the determinants of alcohol abuse in HIV-positive persons using the framework of vector generalized additive models \citep[VGAMs;][]{YeeVGAMBook}. More recent examples include \citet{safari2020}, \citet{bourguignon2021}, \citet{medeiros2021} and \citet{dos2021}. 
The case of {\em correlated} gamma distributed components has earlier been studied by \citet{Lee1979} and \citet{tubbs1986}. Based on Kibble's bivariate gamma distribution for $(U,V)$ \citep{Kibble1941}, \citet{berger2019} developed the \emph{extended GB2 (eGB2) model} for the ratio of two {\it positively correlated} gamma distributed components. Their model is characterized by three parameters, of which one is directly interpretable in terms of the Pearson correlation coefficient between the two components. Conceptually, the extended GB2 model can be seen as a distributional regression model embedded in the framework of generalized additive models for location, scale and shape \citep[GAMLSS;][]{rigsta2005,berger2020}. 

Despite its major importance in biostatistics, no regression modeling strategy exists (to the best of our knowledge) for ratio outcomes with two \emph{negatively correlated} gamma distributed components. Negatively correlated measurements are encountered in numerous applications, for example in dementia research, where ratios of cerebrospinal fluid (CSF) biomarkers are used for the early diagnosis of Alzheimer's disease \citep{koyama}. 
Importantly, measurements of the widely employed amyloid-$\beta$ 42 protein and total tau protein biomarkers are known to exhibit a negative correlation \citep{tapiola2009}. In recent publications, the Gaussian regression model has been used for modeling ratios of CSF biomarkers \citep[e.g.,][]{xu2020}. Clearly, this model neither accounts for the characteristics of the bivariate distribution of $(U,V)$ nor for the skewness in the distribution of the ratio outcome $R$.

Motivated by these problems, and to address the current shortcomings in modeling ratio outcomes with negatively correlated gamma distributed components, we propose a regression model where the joint bivariate distribution of the two gamma distributed components is defined by Frank copula \citep{genest1987}. By this choice, the model flexibly accounts for either negative or positive associations between the two components (measured by Spearman's or Kendall's rank correlation coefficient). It also allows for modeling different characteristics of the two marginal distributions, including possibly unequal rate and shape parameters. By relating the covariates~$X$ to the parameters of the marginals, as well as to the association parameter defined by the copula, our model further allows to derive the conditional probability density function of~\mbox{$R\,|\,X$} as a function of covariates. This, in turn, allows for the analysis of conditional distributional parameters (like the expected value, median or quantiles), including valid inferential conclusions for these quantities. 

We apply the new approach to data from a multi-center observational cohort study conducted by the German Dementia Network  \citep[DCN;][]{kornhuber}. 
Study participants were diagnosed with either mild cognitive impairment (MCI), Alzheimer's Disease (AD), or other dementia. The study aims at determining the diagnostic and prognostic power of clinical, laboratory and imaging methods. This task is considered to be a major challenge, as the period from the first clinical symptoms of AD to disease onset might take years to decades \citep{sperling2013}. Consequently, as biomarker ratios like the amyloid-$\beta$ 42/total tau ratio are considered to be strong predictors of AD progression, it is of high interest to relate these measurements to patient characteristics like age, sex and educational level \citep{jack2015}. As will be demonstrated in Section~\ref{sec:app}, the proposed copula regression model can be suitably applied to address this problem, resulting in meaningful descriptive and inferential findings regarding the associations between the biomarker ratio and individual patient characteristics.

The rest of the paper is organized as follows: Section \ref{sec:method} derives the distributional copula regression model, states novel theoretical results with implications for the interpretation of covariate effects, and presents estimation, prediction and inference. The efficacy of our approach is demonstrated empirically in a simulation study in Section \ref{sec:sim} and in our main application to AD progression in Section \ref{sec:app}. The main findings of the paper are discussed in Section \ref{sec:discussion}.

\section{Methods}\label{sec:method}

Section \ref{sec:method1} starts by deriving the distribution of the ratio of two gamma distributed components with dependence induced by a Frank copula. Details on model specification and fitting are given in Section \ref{sec:definition}. Section \ref{sec:fe_ci} covers the prediction of distributional parameters and inference.

\subsection{Distributional concept}\label{sec:method1} \vspace{.2cm}

Let $U$ and $V$ be two gamma distributed random variables with probability density functions (PDF)
\begin{align}
f_U(u)=\frac{\lambda_U^{\delta_U}}{\Gamma(\delta_U)}u^{\delta_U-1}\exp(-\lambda_Uu) \ \ \ \text{and} \ \ \
f_V(v)=\frac{\lambda_V^{\delta_V}}{\Gamma(\delta_V)}v^{\delta_V-1}\exp(-\lambda_Vv)\,, 
\end{align}
where $\lambda_U$, $\lambda_V$ $>0$ denote the rate parameters and $\delta_U$, $\delta_V$ $>0$ denote the shape parameters of $f_U$ and $f_V$, respectively.  
We allow for dependencies between $U$ and~$V$, and assume that their joint distribution can be described by Frank copula with copula function $C_{\theta}$. By Sklar's theorem, the joint distribution of $(U,V)$ is thus given by
\begin{align}\label{eq:FrankC}
F_{U,V}(u,v)&=C_{\theta}\left(F_U(u),F_V(v)\right)\nonumber \\
&=-\frac{1}{\theta}\,\log\left(1+\frac{(\exp(-\theta F_U(u))-1)(\exp(-\theta F_V(v))-1)}{\exp(-\theta)-1}\right)\,,
\end{align}
where $F_{U,V}$, $F_U$ and $F_V$ denote the joint bivariate and marginal cumulative distribution functions (CDFs) of $U$ and $V$, respectively. The parameter $\theta \in \mathbb{R} \setminus\{0\}$ determines the association between $U$ and $V$. It can be shown that Kendall's rank correlation coefficient $\tau \in [-1,1]$ is a monotone increasing function of $\theta$, given by 
\begin{equation}\label{eq:tau}
\tau (\theta) = 1 \, + \, \frac{4}{\theta} \left(\frac{1}{\theta} \int_0^\theta \frac{t}{e^t - 1} \, dt - 1 \right)
\end{equation}
\citep{joe2014}. As a consequence, the CDF in~\eqref{eq:FrankC} allows for (possibly highly) positive or negative correlations between the two components $U$ and $V$. The joint PDF of $(U,V)$ is given by
\begin{align}\label{eq:Frankc}
f_{U,V}(u,v)&=\frac{\partial^2}{\partial u \partial v}{F_{U,V}(u,v)}=c_{\theta}\big(F_U(u),F_V(v)\big)\,f_U(u)\,f_V(v)&\nonumber \\[.1cm]
&= \frac{-\theta \, \exp(-\theta F_U(u))\, \exp(-\theta F_V(v)) \, (\exp(-\theta)-1) \, f_U(u) \, f_V(v)}{((\exp(-\theta)-1)+(\exp(-\theta F_U(u))-1)\, (\exp(-\theta F_V(v))-1))^2}\,,
\end{align}
where $c_\theta(a,b):=\partial^2 / (\partial a \partial b) \, C_\theta(a,b)$ is the PDF of Frank copula. 

We derive the resulting PDF of the ratio $R$, an interpretable representation thereof and the CDF in the following three propositions.

\begin{proposition}\label{prop1} Let the PDF of $(U,V)$ be defined by \eqref{eq:Frankc}. Then the PDF of the ratio $R:=U/V$ is given by 
\begin{align}\label{eq:fR1}
f_R(r)=\int_{0}^{1}&{\left|F_V^{-1}(s)\right|\,c_\theta\left(F_U(r\,F_V^{-1}(s)),s\right)\,f_U(rF_V^{-1}(s))}\,ds\nonumber \\
=\int_{0}^{1}&\frac{\exp(-\theta\,F_U(r\,F_V^{-1}(s)))\,\exp(-\theta s)(-\theta)(\exp(-\theta)-1)}{((\exp(-\theta)-1)+(\exp(-\theta\,F_U(r\,F_V^{-1}(s)))-1)(\exp(-\theta s)-1))^2} \nonumber \\[.1cm]
&\times F_V^{-1}(s) \, f_U(r\,F_V^{-1}(s))\,ds\,,
\end{align}
where $|\cdot|$ denotes the absolute value function.
\end{proposition}
\begin{proof} Proposition \ref{prop1} can be derived from Proposition 1 of \citet{ly2019}, who provided analytical results for the PDF of the quotient~$U/V$ of two random variables whose dependence structure can be described by an absolutely continuous copula. 
\end{proof}
\begin{figure}[!t]
\begin{center}
\includegraphics[width=0.75\textwidth]{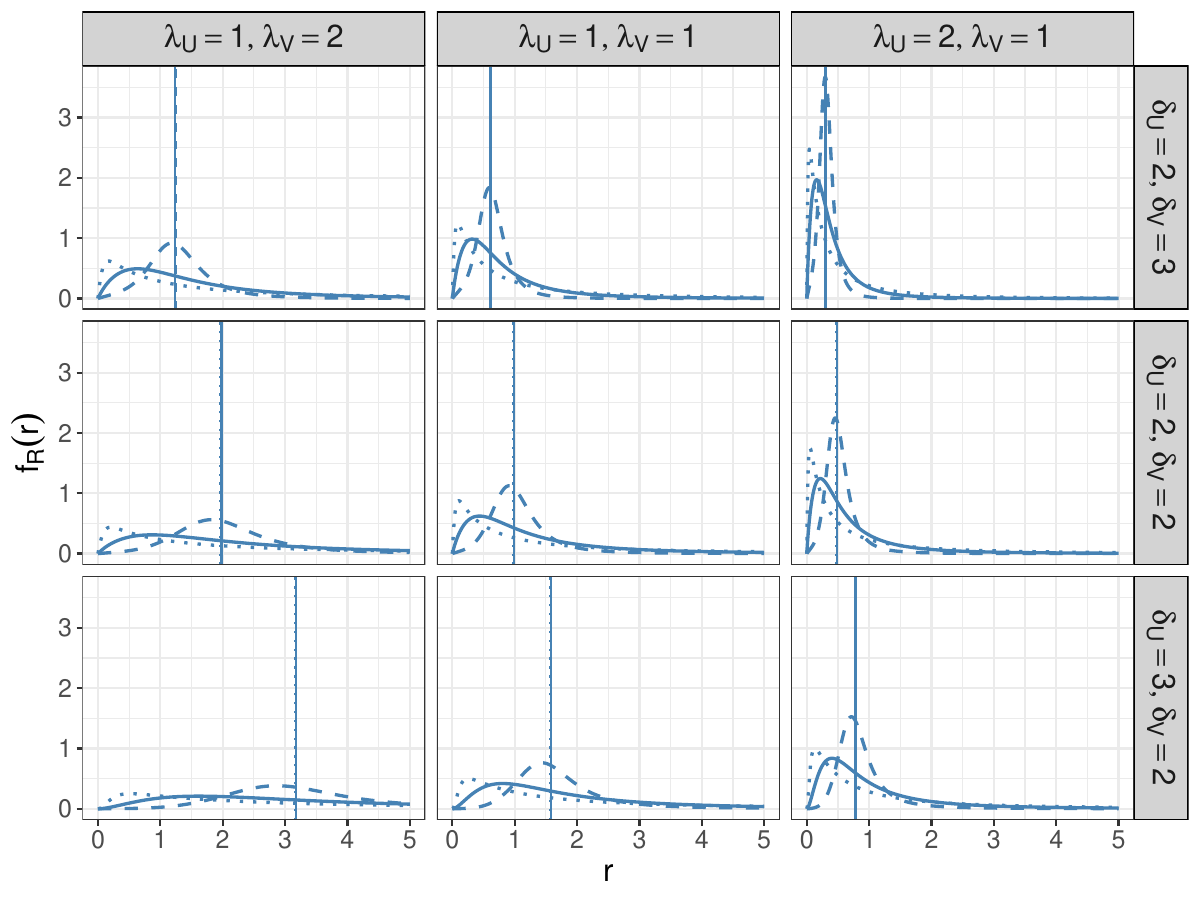}
\caption{Examples of the PDF\@ of $R=U/V$ derived in Proposition \ref{prop1} for parameters $\lambda_U, \lambda_V \in \{1,2\}$, $\delta_U, \delta_V \in \{2,3\}$ and $\theta \in \{-10,1,10\}$ (corresponding to rank correlation coefficients $\tau \in \{-0.67, 0.11, 0.67\}$). In each panel the three lines refer to $\theta=-10$ (dotted), $\theta=1$ (solid) and $\theta=10$ (dashed). Vertical lines refer to the median values of $R$.}
\label{fig:fR}
\end{center}
\end{figure} 
\begin{remark}
Figure~\ref{fig:fR} visualizes the PDF of $R$ for different values of the rate, shape and association parameters. The figure illustrates that the form of the PDF is strongly related to the ratio of  marginal means $\mathbb{E}(U)/\mathbb{E}(V)=\lambda_V\delta_U/\lambda_U\delta_V$, which is highest in the lower left panel ($\mathbb{E}(U)/\mathbb{E}(V)=3$) where the dispersion is very large, and lowest in the upper right panel ($\mathbb{E}(U)/\mathbb{E}(V)=1/3$) where the PDFs are heavily right-skewed. Figure~\ref{fig:fR} also describes the association between the PDF and the Kendall's rank correlation coefficient. In each of the nine cases the mode of the PDF increases as $\theta$ increases. Of note, the median of $R$ does not vary with $\theta$ (being equal for the three PDFs in each panel). 
\end{remark}

\begin{proposition}\label{prop2} Let the PDF of $(U,V)$ be defined by \eqref{eq:Frankc}. Then the PDF of the random variable $R$ in \eqref{eq:fR1} can be re-written as 
\begin{align}\label{eq:fR2}
f_R(r) = \int_{0}^{1}&{c_\theta\left(\frac{1}{\Gamma(\delta_U)}\,\gamma\left(\delta_U, r\,\Lambda\, \gamma^{-1}\left(\delta_V, \Gamma(\delta_V)s\right)\right),s\right)} \nonumber \\
& \times \, \frac{\Lambda^{\delta_U}\, r^{(\delta_U-1)}}{\Gamma(\delta_U)}\left(\gamma^{-1}\left(\delta_V, \Gamma(\delta_V)s\right)\right)^{\delta_U} \, \exp\left(-r\, \Lambda\, \gamma^{-1}\left(\delta_V, \Gamma(\delta_V)s\right)\right)\,ds\,,
\end{align}
where $\Lambda := \lambda_U / \lambda_V$ denotes the ratio of the two rate parameters and $\gamma(\cdot,\cdot)$ is the lower incomplete gamma function. 
\end{proposition}
\noindent \textit{Proof.} The proof of Proposition \ref{prop2} is given in Appendix A.

\begin{remark} By Proposition~\ref{prop2} the PDF of $R$ can be written as a function of the ratio of the rate parameters $\Lambda = \lambda_U / \lambda_V$. This facilitates the interpretation of the proposed regression model, as it results in a sparser representation of the covariate effects. In particular, covariate effects can be investigated using one-dimensional hypothesis tests and p-values, see our application in Section~\ref{sec:app}. Figure~\ref{fig:med} illustrates how the median of $R$ is related to~$\Lambda$. It is seen that the median decreases with $\Lambda$ independent of the values of the shape parameters and the association parameter.
\end{remark}

\begin{proposition}\label{prop3}
Let the CDF of $(U,V)$ be defined by \eqref{eq:FrankC}. Then the CDF of the random variable $R$ is given by 
\begin{align}\label{eq:cdfr}
F_R(r)=\int_{0}^{1}\, {\frac{A\,f_U(r\,F_V^{-1}(s))\,\frac{r}{f_V(F_V^{-1}(s))} \, (\exp(-\theta s)-1)+(A-1)\exp(-\theta s)}{\exp(-\theta) +(A-1)\exp(-\theta s)- A }}\,ds\,,
\end{align}
where $A=\exp\left(-\theta\,F_U(rF_V^{-1}(s))\right)$.
\end{proposition} 

\begin{proof} By Equation~(9) of \citet{ly2019}, the CDF of $R$ is derived as
\begin{align*}
F_R(r) \, & \, =\underbrace{F_V(0)}_{=0} \, + \, \int_{0}^{1}{\underbrace{\text{sgn}(F_V^{-1}(s))}_{=1} \, \frac{\partial}{\partial s} C_\theta\left(F_U(r\,F_V^{-1}(s)),s\right)}\,ds \\
&\,=-\frac{1}{\theta}\, \int_{0}^{1}{\frac{\partial}{\partial s}\log\left( 1+ \frac{\big(\exp (-\theta\,F_U(rF_V^{-1}(s)))-1\big) \, (\exp(-\theta s)-1)}{\exp(-\theta)-1}\right)}\,ds \\ 
&\, =\int_{0}^{1} \left( {\frac{\exp (-\theta\,F_U(rF_V^{-1}(s)) )\,f_U(r\,F_V^{-1}(s))\, \frac{r}{f_V(F_V^{-1}(s))} \, (\exp(-\theta s)-1)}{\exp(-\theta)-1+\big(\exp (-\theta\,F_U(rF_V^{-1}(s)) )-1 \big) \, (\exp(-\theta s)-1)}}\right.\nonumber \\[.1cm]
&\hspace{0.85cm}\left.+ \frac{ \big( \exp (-\theta\,F_U(rF_V^{-1}(s)) )-1 \big) \, \exp(-\theta s)}{\exp(-\theta)-1+\big(\exp (-\theta\,F_U(rF_V^{-1}(s)) )-1\big) \, (\exp(-\theta s)-1)}\right) ds\,,
\end{align*}
where $\text{sgn}(\cdot)$ is the sign function. Rearrangement of the last equation gives \eqref{eq:cdfr}.
\end{proof}

\begin{figure}[!t]
\begin{center}
\includegraphics[width=0.6\textwidth]{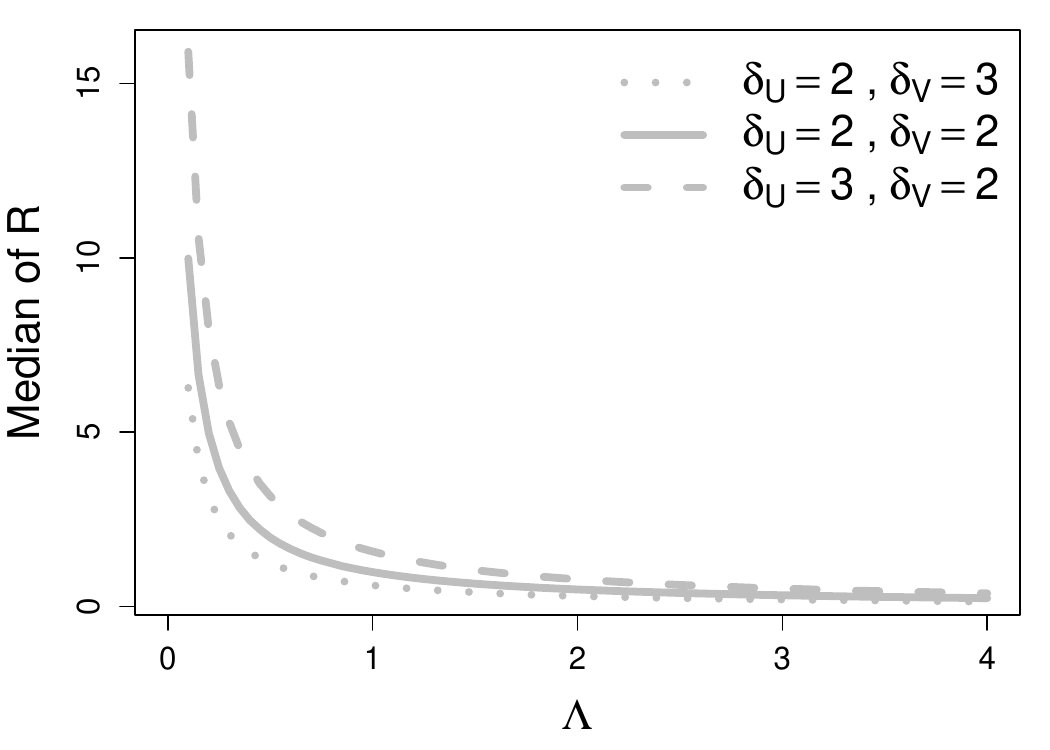}
\caption{Median of $R=U/V$ for parameters $\Lambda \in [0.1, 4]$ and $\delta_U, \delta_V \in \{2,3\}$, as calculated from the formula in Proposition~\ref{prop2}. Note that the median of $R$ does not vary with $\theta$.}
\label{fig:med}
\end{center}
\end{figure} 

\subsection{Regression specification and estimation}\label{sec:definition}  \vspace{.2cm}

To model the entire distribution of $R$ as a function of covariates $X=(X_1,\hdots,X_p)^\top$, we propose to relate both  the logarithmic rate parameters $\lambda_U$ and~$\lambda_V$ and the association parameter $\theta$ to predictor functions of the form
\begin{align}
\label{U}\log(\lambda_U|X)&=\eta_U=\beta_{U0}+\beta_{U1}X_1+\hdots+\beta_{Up}X_p\,,\\
\label{V}\log(\lambda_V|X)&=\eta_V=\beta_{V0}+\beta_{V1}X_1+\hdots+\beta_{Vp}X_p\quad \text{and}\\
\label{theta}\theta|X&=\eta_\theta=\beta_{\theta 0}+\beta_{\theta 1}X_1+\hdots+\beta_{\theta p}X_p\,,
\end{align}
where $\bs{\beta}_U=(\beta_{U0},\hdots,\beta_{Up})^\top$, $\bs{\beta}_V=(\beta_{V0},\hdots,\beta_{Vp})^\top$ and  $\bs{\beta}_\theta =(\beta_{\theta 0},\hdots,\beta_{\theta p})^\top$ are sets of real-valued coefficients. Analogous to classical gamma regression (Chapter 5.3 of \citealp{fahrmeir2022}), the use of the logarithmic transformation in~\eqref{U} and~\eqref{V} ensures positivity of the rate parameters. Since $\theta\in\mathbb{R} \setminus\{0\}$, no transformation is needed for the association parameter. As a result of~\eqref{U} and~\eqref{V} it holds that $\log(\Lambda|X)=\eta_U - \eta_V$.

\begin{remark}
    In principle, our approach allows to make use of the full flexibility of GAMLSS by relating all distributional parameters (including the shape parameters $\delta_U,\delta_V$) to the covariates and by including nonlinear effects in the predictor functions. However, in our application we found that the specification in \eqref{U} -- \eqref{theta} provides a sufficient fit, thereby meeting a compromise between model fit and model complexity. Furthermore, it greatly simplifies the interpretability of the results (as we will further elaborate in Section~\ref{sec:app}). Based on these considerations, our model assumes that the shape parameters $\delta_U$ and~$\delta_V$ do not depend on $X$, but can be treated as nuisance parameters.
\end{remark}
\begin{defin}
The regression model for the ratio $R=U/V$ with the distribution from Proposition \ref{prop3} and with covariate-dependent parameters as specified in \eqref{U}--\eqref{theta} will be termed \textit{Frank copula with Gamma Distributed Marginals (FCGAM)} in the following. The FCGAM model imposes the constraint $\delta_U$, $\delta_V$ $>1$ to ensure that the two marginals both exhibit a unimodal, right-skewed distribution, which is the common form of biomarker distributions in medical applications. 
\end{defin}

\begin{cor}
For a set of i.i.d. observations $(u_1,v_1,\bs{x}_1^\top)^\top,\ldots,(u_n,v_n,\bs{x}_n^\top)^\top$ with ratios $r_1=u_1/v_1,\ldots,r_n=u_n/v_n$ and model coefficients $\bs{\gamma}=\left(\bs{\beta}_U^\top,\bs{\beta}_V^\top,\bs{\beta}_\theta^\top,\delta_U,\delta_V\right)^\top$, the log-likelihood function of the FCGAM model is given by 
\begin{align}
\label{eq:loglik}
\ell(\bs{\beta}_U,\bs{\beta}_V,&\bs{\beta}_\theta,\delta_U,\delta_V;u_1,\hdots,u_n,v_1,\hdots,v_n,\bs{x}_1,\hdots,\bs{x}_n)\nonumber \\
=\sum_{i=1}^n\,\bigg\{&\log\Big(f_{U,V}(u_i,v_i|\bs{x}_i,\bs{\beta}_U,\bs{\beta}_V,\bs{\beta}_\theta,\delta_U,\delta_V)\Big)\bigg\}\nonumber\\
=\sum_{i=1}^n\,\bigg\{&\log\Big(\exp(-\bs{x}_i^\top\bs{\beta}_\theta\,F_U(u_i;\bs{x}_i^\top\bs{\beta}_U,\delta_U))\exp(-\bs{x}_i^\top\bs{\beta}_\theta\,F_V(v_i;\bs{x}_i^\top\bs{\beta}_V,\delta_V))\nonumber\\
&\ \times \ (-\bs{x}_i^\top\bs{\beta}_\theta)(\exp(-\bs{x}_i^\top\bs{\beta}_\theta)-1)f_U(u_i;\bs{x}_i^\top\bs{\beta}_U,\delta_U)\,f_V(v_i;\bs{x}_i^\top\bs{\beta}_V,\delta_V)\Big)\nonumber\\
-2\,&\log \,\Big((\exp(-\bs{x}_i^\top\bs{\beta}_\theta)-1)+(\exp(-\bs{x}_i^\top\bs{\beta}_\theta\,F_U(u_i;\bs{x}_i^\top\bs{\beta}_U,\delta_U))-1)\nonumber\\
&\ \times \ (\exp(-\bs{x}_i^\top\bs{\beta}_\theta\,F_V(v_i;\bs{x}_i^\top\bs{\beta}_V,\delta_V))-1)\Big)\bigg\}\,.
\end{align}
\end{cor}
\begin{cor}\label{cor2}
Under the usual regularity assumptions, the estimator
\begin{align}
\hat{\bs{\gamma}}=(\hat{\bs{\beta}}^\top_U,&\hat{\bs{\beta}}^\top_V,\hat{\bs{\beta}}^\top_\theta,\hat{\delta}_U,\hat{\delta}_V)^\top := \nonumber \\
&\underset{\bs{\gamma}=\bs{\beta}_U,\bs{\beta}_V,\bs{\beta}_\theta,\delta_U,\delta_V}{\text{argmax}}\ell(\bs{\beta}_U,\bs{\beta}_V,\bs{\beta}_\theta,\delta_U,\delta_V;u_1,\hdots,u_n,v_1,\hdots,v_n,\bs{x}_1,\hdots,\bs{x}_n)\,
\end{align}
is consistent and asymptotically normal for $n\to\infty$.
\end{cor}

\noindent {\it Implementational details.} Maximization of the log-likelihood function in~\eqref{eq:loglik} can be carried out using the~R function \texttt{FCGAMoptim()}, which is part of the supplementary material to this paper. The optimization algorithm is based on the Broyden-Fletcher-Goldfarb-Shanno (BFGS) algorithm implemented in the R function \texttt{optim()}, setting the additional constraint $\delta_U$, $\delta_V$ $>1$. 

\subsection{Prediction of distributional parameters and inference}\label{sec:fe_ci}  \vspace{.2cm}

{\it Prediction.} For a new observation with covariate values $\tilde{\bs{x}}$, predictions of the conditional PDF ${f}_R(r|\tilde{\bs{x}})$ can be obtained by computing the maximum likelihood estimate (MLE) and by plugging the estimated parameters $\hat{\Lambda}\,|\, \tilde{\bs{x}}=\exp(\tilde{\bs{x}}^\top\hat{\bs{\beta}}_U-\tilde{\bs{x}}^\top\hat{\bs{\beta}}_V)$, $\hat{\theta}\,|\,\tilde{\bs{x}}=\tilde{\bs{x}}^\top\hat{\bs{\beta}}_\theta$ and $\hat{\delta}_U$, $\hat{\delta}_V$ in Equation~\eqref{eq:fR2}. The predicted PDF can then be used to predict any distributional parameter of interest (like the expected value, median or quantiles). For example, denoting the predicted PDF by $\hat{f}_R(r|\tilde{\bs{x}})$, the predicted median can be calculated by \begin{align}\label{eq:medR}
\hat{r}_{\text{med}} \,|\, \tilde{\bs{x}}  = \text{min}\bigg\{r\in \mathbb{R}^+ \, \Big| \, \int_{0}^{r}\hat{f}_R(s|\tilde{\bs{x}})\, ds \geq 0.5\bigg\}\,.
\end{align}

\noindent {\it Inference.} Despite the asymptotic results from Corollary \ref{cor2}, more reliable finite-sample confidence intervals have been established in additive models \citep{wood2017}. This is particularly the case for the quantities of interest here (such as the median of $R$ above). The reason is that these are nonlinear transformations of the original model coefficients such that confidence intervals would show an additional finite-sample bias due to the application of the $\Delta$-rule. 
Following \citet{wood2017}, we thus propose to construct confidence intervals of $\bs{\gamma}$ using a Bayesian approach, which we accordingly refer to as credible intervals. Assuming flat priors on $\bs{\gamma}$, the posterior distribution of $\bs{\gamma}$ is given by
\begin{align}
\label{eq:posterior}
\bs{\gamma} \, | \, u_1,\hdots,u_n,v_1,\hdots,v_n \, \sim \, N\big(\hat{\bs{\gamma}},\hat{J}^{-1}(\hat{\bs{\gamma}})\big)\,,
\end{align}
where $\hat{J}(\hat{\bs{\gamma}})$ is the Hessian of the negative log-likelihood evaluated at $\hat{\bs{\gamma}}$ (Equation~(6.26) of \citealp{wood2017}). 
Consequently, approximate $(1-\alpha)\%$ credible intervals for the coefficients $\bs{\gamma}$ can be obtained by drawing a large sample from the posterior distribution \eqref{eq:posterior} and by calculating the $\alpha/2$ and $(1-\alpha/2)$ percentiles from this sample \citep[][p.~293]{wood2017}. In our simulations (Section \ref{sec:sim}) and in the analysis of the DCN study data (Section \ref{sec:app}) we used samples of size 10,000 throughout.

\section{Simulations}
\label{sec:sim}

We conducted three simulation studies to investigate the performance of the FCGAM model. Our main aims were (a) to analyze the model fit and the coverage  of the credible intervals, (b) to evaluate how the performance of the FCGAM approach is affected by the sample size and the choice of the association parameter~$\theta$, and (c) to benchmark our method against alternative ones, in particular against the extended GB2 model by \cite{berger2019} which assumes the correlation between $U$ and $V$ to be positive.

\subsection{Experimental design}  \vspace{.2cm}

In all simulations the ratio outcome was generated according to the PDF of the FCGAM model derived in Proposition \ref{prop1}. Similar to the application data in Section~\ref{sec:app}, we considered two standard normally distributed covariates $X_1,X_2 \sim N(0,1)$ and two binary covariates $X_3,X_4 \sim B(1, 0.5)$, which were equi-correlated with Pearson correlation coefficient 0.4. For each $n \in \{200, 500, 1000\}$ we simulated $1000$ independent data sets. 

In \textit{Simulation Study~1}, we considered scenarios with fixed {\em negative} correlation (the case which motivated our development of the FCGAM model), setting $\beta_{\theta0} \in \{-1, -5, -10\}$ and $\beta_{\theta1}=\hdots=\beta_{\theta4}=0$. This resulted in the respective rank correlation coefficients $\tau \in \{-0.11, -0.46, -0.67\}$. The rate parameters were related to the four covariates through the coefficients $\beta_U=(0,0.4,-0.4,0.2,-0.2)^\top$ and $\beta_V=(0,-0.2, 0.2, -0.4, 0.4)^\top$. The shape parameters were set to $\delta_U=2$ and $\delta_V=6$ in all scenarios. 

In \textit{Simulation Study~2}, we considered scenarios with fixed {\em positive} correlation (the case which has already been covered by the eGB2 model but also applies to the FCGAM model), setting $\beta_{\theta0} \in \{1, 5, 10\}$ and $\beta_{\theta1}=\hdots=\beta_{\theta4}=0$. This resulted in the respective rank correlation coefficients $\tau \in \{0.11, 0.46, 0.67\}$. To ensure that the outcome values were in a meaningful range (comparable to \textit{Simulation Study~1}) we set the regression coefficients to $\beta_U=(0,0.4,-0.4,0.2,-0.2)^\top$ and $\beta_V=(0,0.2,-0.2,0.4,-0.4)^\top$, and the shape parameters to $\delta_U=2$ and $\delta_V=2$.

In \textit{Simulation Study 3}, we evaluated how the model fit of the FCGAM model was affected when falsely assuming a dependence of $\theta$ on $X_1,\hdots,X_4$, or when ignoring a present dependence of $\theta$ on $X_1,\hdots,X_4$. For this we reconsidered the data sets from \textit{Simulation Study~1} with $\tau= -0.11$, and additionally considered scenarios where  the association parameter $\theta$ was related to the four covariates through the coefficient vector $\beta_\theta = (0, 1, -1, 0.5, -0.5)^\top$ (resulting in covariate-dependent rank correlation coefficients $\tau_i$, 
with the remaining parameters as in \textit{Simulation Study~1}). In both cases we fitted the FCGAM model with covariate-dependent $\theta$ (according to~\eqref{theta}) and with constant $\theta = \beta_{\theta0}$. \\

\noindent {\it Benchmark methods.} We evaluated the fits of the 1000 FCGAM models by computing the predictive log-likelihood values on 1000 independent test data sets. The test data sets (of size $n$ each) were also used to compare the FCGAM model to alternative models. To this purpose, we evaluated the predictive log-likelihood values of the following benchmark methods, where (ii), (iii) and (vi) are univariate regression models for $R$, (iv) is a univariate regression model for $\log(R)$, and (v) and (vii) are distributional regression models:\\

\begin{minipage}{0.1\textwidth}

\end{minipage}
\begin{minipage}{0.96\textwidth}
\begin{itemize}
\setlength{\itemsep}{-0.25em}
    \item[(i)] The copula-based \textit{FCGAM} model introduced in Section \ref{sec:definition}.
    \item[(ii)] The extended GB2 model \citep{berger2019} (\textit{eGB2}) assuming a positive correlation between $U$ and~$V$. 
    \item[(iii)] The simple GB2 model (\textit{GB2}) assuming zero Pearson correlation between~$U$ and~$V$.
    \item[(iv)] A  Gaussian regression model with log-transformed outcome values (\textit{LN}). 
    \item[(v)] A Gaussian GAMLSS with log-transformed outcome values, where both the mean and the standard deviation were related to the covariates (\textit{LN.LSS}). The standard deviation was modeled using the \mbox{log link}.
    \item[(vi)] A Gamma regression model with the original outcome values (\textit{GA}). The mean parameter was related to the covariates and was modeled using the log link.
    \item[(vii)] A Gamma GAMLSS with the original outcome values, where both the mean and the scale parameters were related to the covariates (\textit{GA.LSS}) using the log link. 
\end{itemize}
\end{minipage}

\subsection{Results} \vspace{.2cm}

{\it Point estimates of the FCGAM coefficients.}
Figure~\ref{fig:sim1:estx} presents the coefficient estimates $\hat{\beta}_U$ in \textit{Simulation Study~1} with negative (but covariate-independent) correlation between $U$ and $V$. The boxplots show that on average the estimated coefficients are very close to the true ones, regardless of the association parameter~$\theta$. Accordingly, the finite-sample bias of the MLEs is small in all scenarios (with varying $n$ and $\theta$). From Figure~\ref{fig:sim1:estx} it can also be seen that, as expected, the variance of the estimates decreases with increasing sample size, in particular for the two binary covariates $X_3$ and $X_4$. In contrast, the correlation (determined by the value of~$\theta$) has only a small impact on the variance of the estimates. The coefficient estimates $\hat{\beta}_V$ (presented in~Supplementary Figure~S1) exhibit even smaller variances in the scenarios with $n=500$ and $n=1000$. 

The coefficient estimates $\hat{\beta}_U$ and $\hat{\beta}_V$ from \textit{Simulation Study 2} with positive correlation between $U$ and $V$ are shown in Supplementary Figures~S2 and~S3, respectively. In both cases the bias is small throughout all scenarios. Regarding the variance of the estimates, the results are largely the same as in Figure~\ref{fig:sim1:estx}. \\

\begin{figure}[!t]
\begin{center}
\includegraphics[width=0.8\textwidth]{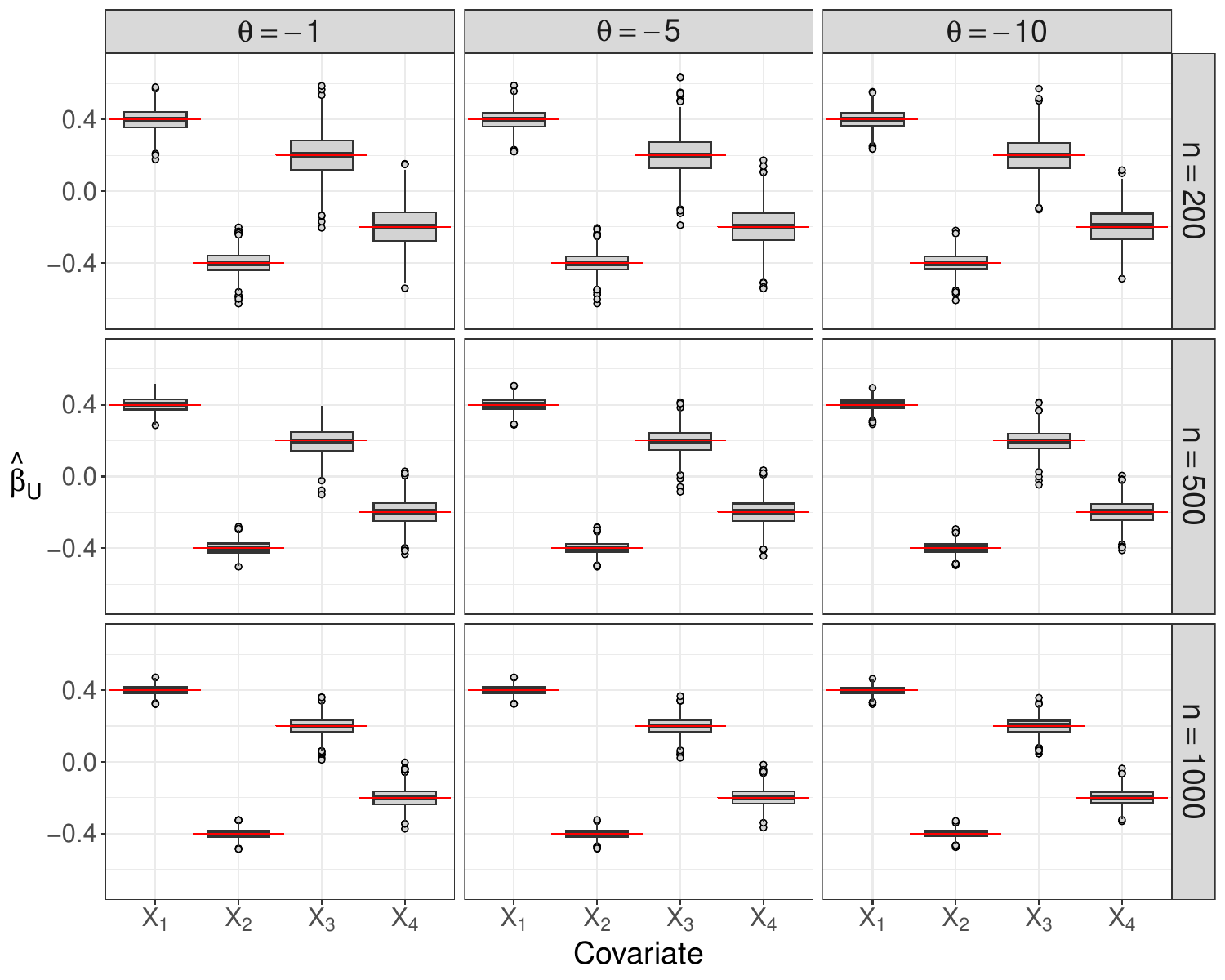}
\caption{Point estimates of the FCGAM coefficients in \textit{Simulation Study 1}. The boxplots visualize the MLEs of the coefficients $\beta_{U1}=0.4$, $\beta_{U2}=-0.4$, $\beta_{U3}=0.2$ and $\beta_{U4}-0.2$ that were obtained from fitting the FCGAM model to 1000 data sets of size $n$ each. The red lines refer to the true values of the coefficients.}
\label{fig:sim1:estx}
\end{center}
\end{figure} 

\noindent {\it Coverage proportions of the credible intervals.} The coverage proportions of the 95\% credible intervals obtained from the FCGAM fits are presented in Table~\ref{tab:sim:coverage}. They range between $0.928 - 0.958$ (\textit{Simulation Study 1}) and between $0.928 - 0.962$ (\textit{Simulation Study 2}), which is close to the nominal coverage of 95\%. There were only minor differences with regard to sample size and the correlation coefficient. This result demonstrates that not only point estimation but also inference works well for highly positive or negative correlations and fairly small samples. \\

\begin{table}[!t]
\caption{Coverage proportions of the FCGAM credible intervals. For each coefficient $\beta_{Uj}$, $j=1,\hdots,4,$ and $\beta_{Vj}$, $j=1,\hdots,4,$ the table contains the coverage proportion of the 95\% credible interval, as obtained from fitting the FCGAM model to 1000 independent data sets of size $n$ each. 
}
\begin{center}
\begin{footnotesize}
\begin{tabular}{llcccccccc}
\toprule
\multicolumn{2}{c}{\textbf{Simulation Study 1}}&$\beta_{U1}$&$\beta_{U2}$&$\beta_{U3}$&$\beta_{U4}$&$\beta_{V1}$&$\beta_{V2}$&$\beta_{V3}$&$\beta_{V4}$\\
\midrule
$n$=200&$\theta=-1$&0.938&0.932&0.936&0.938&0.949&0.949&0.954&0.950\\
&$\theta=-5$&0.958&0.942&0.955&0.941&0.943&0.935&0.946&0.952\\
&$\theta=-10$&0.937&0.938&0.932&0.947&0.952&0.946&0.940&0.949\\
\midrule
$n$=500&$\theta=-1$&0.949&0.937&0.940&0.938&0.952&0.953&0.946&0.930\\
&$\theta=-5$&0.954&0.938&0.928&0.935&0.952&0.932&0.935&0.949\\
&$\theta=-10$&0.934&0.948&0.939&0.951&0.955&0.947&0.958&0.943\\
\midrule
$n$=1000&$\theta=-1$&0.949&0.937&0.955&0.947&0.944&0.946&0.948&0.937\\
&$\theta=-5$&0.947&0.942&0.936&0.936&0.957&0.950&0.950&0.950\\
&$\theta=-10$&0.933&0.945&0.934&0.953&0.943&0.948&0.949&0.943\\
\midrule \\[-.3cm]
\midrule
\multicolumn{2}{c}{\textbf{Simulation Study 2}}&$\beta_{U1}$&$\beta_{U2}$&$\beta_{U3}$&$\beta_{U4}$&$\beta_{V1}$&$\beta_{V2}$&$\beta_{V3}$&$\beta_{V4}$\\
\midrule
$n$=200&$\theta=1$&0.935&0.929&0.934&0.935&0.947&0.933&0.959&0.954\\
&$\theta=5$&0.952&0.940&0.954&0.941&0.949&0.951&0.954&0.947\\
&$\theta=10$&0.938&0.939&0.939&0.947&0.951&0.944&0.948&0.953\\
\midrule
$n$=500&$\theta=1$&0.941&0.948&0.936&0.928&0.940&0.942&0.946&0.940\\
&$\theta=5$&0.948&0.941&0.929&0.929&0.952&0.944&0.941&0.962\\
&$\theta=10$&0.931&0.947&0.937&0.951&0.948&0.944&0.949&0.937\\
\midrule
$n$=1000&$\theta=1$&0.950&0.944&0.954&0.944&0.952&0.939&0.953&0.935\\
&$\theta=5$&0.955&0.942&0.934&0.934&0.956&0.952&0.949&0.947\\
&$\theta=10$&0.928&0.946&0.934&0.955&0.935&0.956&0.948&0.945\\
\bottomrule
\end{tabular}
\end{footnotesize}
\end{center}
\label{tab:sim:coverage}
\end{table}

\noindent {\it Distributional prediction.} The root mean squared error (RMSE) of the estimated conditional median values computed from \eqref{eq:medR} are given in Table~\ref{tab:sim:median}. In \textit{ Simulation Study 1} the performance is quite similar for all three values of $\theta$, whereas in \textit{Simulation Study~2} the RMSE considerably decreases with increasing value of $\theta$. This indicates that estimating the median value works best for highly positive correlations where the PDF of $R$ is rather diffuse with a large mode value (compare Figure~\ref{fig:fR}). It is also seen from Table~\ref{tab:sim:median} that the means and the standard deviations of the RMSE decrease with increasing sample size.\\

\begin{table}[!t]
\caption{RMSE of the estimated conditional median values of $R$. The table presents the mean RMSE of the estimated conditional median of $R$, as obtained from fitting the FCGAM model to 1000 independent data sets of size $n$ each. Standard deviations of the RMSE values (across the 1000 data sets) are given in brackets. 
}
\begin{center}
\begin{footnotesize}
\begin{tabular}{llccc}
\toprule
\multicolumn{2}{c}{\textbf{Simulation Study 1}} &$\theta=-1$&$\theta=-5$&$\theta=-10$\\
\midrule
$n$ = 200&&0.076 (0.035)&0.084 (0.039)&0.081 (0.037)\\
$n$ = 500&&0.049 (0.021)&0.053 (0.023)&0.051 (0.021)\\
$n$ = 1000&&0.034 (0.015)&0.038 (0.016)&0.036 (0.016)\\
\midrule \\[-.3cm]
\midrule
\multicolumn{2}{c}{\textbf{Simulation Study 2}} &$\theta=1$&$\theta=5$&$\theta=10$\\
\midrule
$n$ = 200&&0.154 (0.056)&0.101 (0.035)&0.063 (0.021)\\
$n$ = 500&&0.097 (0.033)&0.064 (0.022)&0.040 (0.014)\\
$n$ = 1000&&0.069 (0.023)&0.046 (0.015)&0.029 (0.010)\\
\bottomrule
\end{tabular}
\end{footnotesize}
\end{center}
\label{tab:sim:median}
\end{table}

\begin{figure}[!t]
\begin{center}
\includegraphics[width=0.8\textwidth]{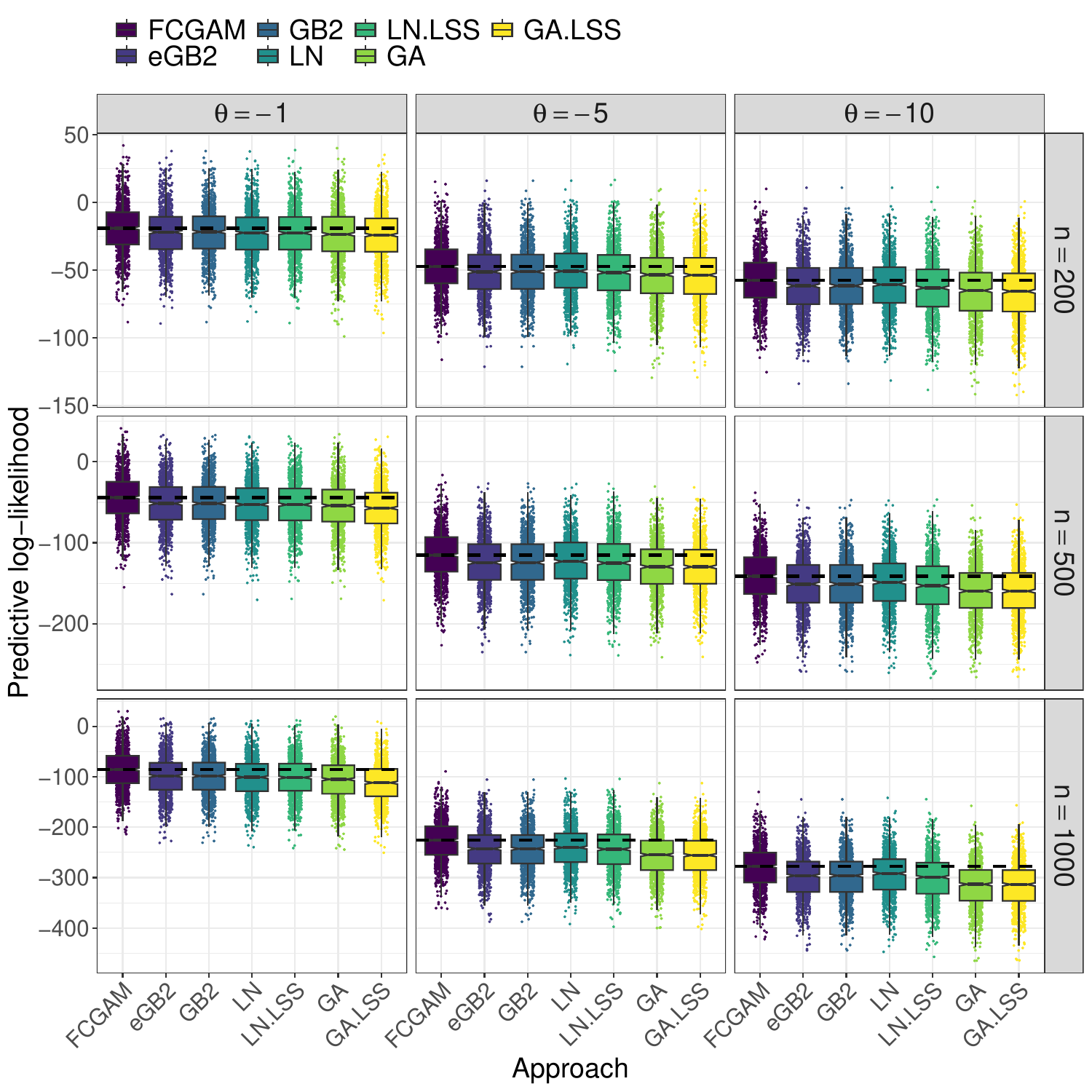}
\caption{Comparison of the FCGAM model to alternative methods in \textit{Simulation Study 1}. The boxplots visualize the predictive log-likelihood values obtained from the FCGAM model and from the benchmark methods (ii) to (vii). All models were fitted to 1000 independent data sets and evaluated on independently generated test data sets of the same size. In each panel, the dashed horizontal line indicates the median predictive log-likelihood value of the best performing method.}
\label{fig:sim1:ll}
\end{center}
\end{figure} 

\noindent {\it Comparison to alternative models.} Figure~\ref{fig:sim1:ll} and Supplementary Figure~S4 show the prediction accuracy (i.e.~the predicted log-likelihood values on the test sets) of the FCGAM model and the benchmark methods (ii) to (vi). In \textit{Simulation Study~1} with negative correlation, it can be observed that the FCGAM model achieves the highest accuracy in all scenarios (Figure~\ref{fig:sim1:ll}). The superiority is even more evident when the sample size and the value of the correlation coefficient are increased. The extended GB2 and simple GB2 methods yield similar performances as the Gaussian models with log-transformed outcome (LN and LN.LSS), whereas the Gamma regression models (GA and GA.LSS) result in the lowest accuracy. For both LN and GA the GAMLSS models are not superior to their simple counterparts.

In \textit{Simulation Study~2} with positive correlation, the results change considerably (Supplementary Figure~S4).  As expected, the performance of the FCGAM and eGB2 models is largely the same, as the eGB2 model also assumes gamma distributed components with positive correlation. The simple GB2 model (assuming uncorrelated components) and the Gaussian models with log-transformed outcomes (LN and LN.LSS) perform comparably well in the scenarios with $\theta=1$, but deteriorated with increasing correlation ($\theta=5$ and $\theta=10$). Again, the Gamma regression models (GA and GA.LSS) exhibit the worst performance by far.\\

\begin{table}[!t]
\caption{Analysis of misspecified models for the association parameter in Simulation Study~3. The table presents the mean RMSE of the estimated conditional median values (upper part) and the mean of the predictive log-likelihood values (lower part), as obtained from fitting the FCGAM model to 1000 independent data sets and evaluating the fits on 1000 independently generated test data sets. Standard deviations (across the 1000 data sets) are given in brackets. The left part of the table refers to the scenarios with covariate-dependent correlation coefficients (in the observed range $\tau_i \in [-0.483,\hdots,0.464]$), whereas the right part refers to the scenarios with fixed negative correlation $\tau=-0.11$. The terms ``modeled $\theta$'' and ``constant $\theta$'' refer to the FCGAM models with a covariate-dependent predictor function for~$\theta$ (as in \eqref{theta}) and an intercept-only predictor function for~$\theta$ ($\beta_{\theta1}= \ldots = \beta_{\theta p}=0$ in \eqref{theta}), respectively.}
\begin{center}
\begin{footnotesize}
\begin{tabular}{llcccc}
\toprule
\multicolumn{2}{l}{\bf RMSE\ \ }&\multicolumn{2}{c}{Covariate-dependent correlation}&\multicolumn{2}{c}{Fixed correlation}\\[.1cm]
&&modeled $\theta$&constant $\theta$&modeled $\theta$&constant $\theta$\\
\midrule
$n$ = 200&&0.069 (0.031)&0.082 (0.042)&0.079 (0.032)&0.077 (0.034)\\
$n$ = 500&&0.042 (0.018)&0.058 (0.026)&0.050 (0.020)&0.047 (0.018)\\
$n$ = 1000&&0.031 (0.013)&0.049 (0.021)&0.036 (0.015)& 0.034 (0.014)\\
\midrule \\[-.3cm]
\midrule
\multicolumn{2}{l}{\bf Predictive }&\multicolumn{2}{c}{Covariate-dependent correlation}&\multicolumn{2}{c}{Fixed correlation}\\[.1cm]
\multicolumn{2}{l}{\bf log-likelihood}& modeled $\theta$&constant $\theta$&modeled $\theta$&constant $\theta$\\
\midrule
$n$ = 200&&-6.769 (19.176)&-7.116 (19.133)&-20.165 (19.163)&-19.356 (19.056)\\
$n$ = 500&&-10.736 (29.031)&-12.853 (29.087)&-44.585 (29.092)&-43.979 (29.050)\\
$n$ = 1000&&-18.237 (41.177)&-23.330 (40.965)&-85.656 (40.989)&-85.071 (40.967)\\
\bottomrule
\end{tabular}
\end{footnotesize}
\end{center}
\label{tab:sim3}
\end{table}

\noindent {\it Misspecified models for the association parameter in Simulation Study 3.} The RMSE of the estimated conditional median values and the predictive log-likelihood values of the FCGAM fits are summarized in Table~\ref{tab:sim3}. It is seen that ignoring the dependence of $\theta$ on the covariates (left part of Table~\ref{tab:sim3}) decreases both the predictive ability and the model fit.
In the scenario with $n=1000$ (large sample size), the difference in predictive log-likelihood values of 5.093 suggests ``considerably less'' empirical support for the model with constant $\theta$ (according to the rules of thumb provided in \citealp{burnham}).  On the other hand, when unnecessarily modeling the dependence of $\theta$ on $X_1,\hdots,X_4$ (right part of Table~\ref{tab:sim3}) the predictive ability and the model fit are mostly unaffected (showing only negligible differences in the RMSE and the predictive log-likelihood values). \\

\noindent {\it Overall summary.} Taken together, we make the following key empirical observations:
\begin{enumerate}
\item Point estimates from the FCGAM model are reliable and nearly unbiased even for small sample sizes.
\item The FCGAM model outperforms the eGB2 model in case of negative correlation and is en par with the eGB2 model when the correlation is positive.
\item In all scenarios, the Gamma regression models perform worst. In particular, they perform worse than the Gaussian models with log-transformed outcome.
\item Falsely modeling the association parameter does not deteriorate predictive performance to a large degree, whereas the FCGAM model with covariate-dependent $\theta$ improves the fit when the true association depends on the covariates.
\end{enumerate}

\section{Cohort Study of the German Dementia Competence Network}
\label{sec:app}

\noindent {\it Background.} The multi-center cohort study conducted by the German Dementia Competence Network \citep[DCN;][]{kornhuber} enrolled patients aged  older than 50 years that were diagnosed with either mild cognitive impairment (MCI), Alzheimer's disease (AD) or other dementia. Recruitment took place between 2003 and 2007. The main objective of the original study was to establish biomarkers for the diagnosis and prognosis of AD using clinical, laboratory and imaging measurements. Here, we investigate covariates that are potentially associated with amyloid-$\beta$ 42, amyloid-$\beta$ 40 and total tau protein concentrations measured in cerebrospinal fluid samples. These analyses are of high relevance for clinical routine in the neurosciences, since biomarkers enable the detection of AD pathology long before the occurrence of the first clinically obvious symptoms \citep{sperling2013}. Thus, relating covariates to biomarker values provides insight into disease pathology and prevention at the individual patient level. In the neurosciences, amyloid-$\beta$ 42, amyloid-$\beta$ 40 and total tau protein concentrations are usually not analyzed separately but in terms of their ratios. More specifically, the amyloid-$\beta$ 42/40 ratio and amyloid-$\beta$ 42/total tau ratio are considered to be strong predictors of AD progression \citep{koyama}. Therefore we focus on the group of MCI patients and relate their ratios to patient-related risk factors for dementia. \\

\begin{figure}[!t]
\begin{center}
\subfloat[][]{\includegraphics[width=0.45\textwidth]{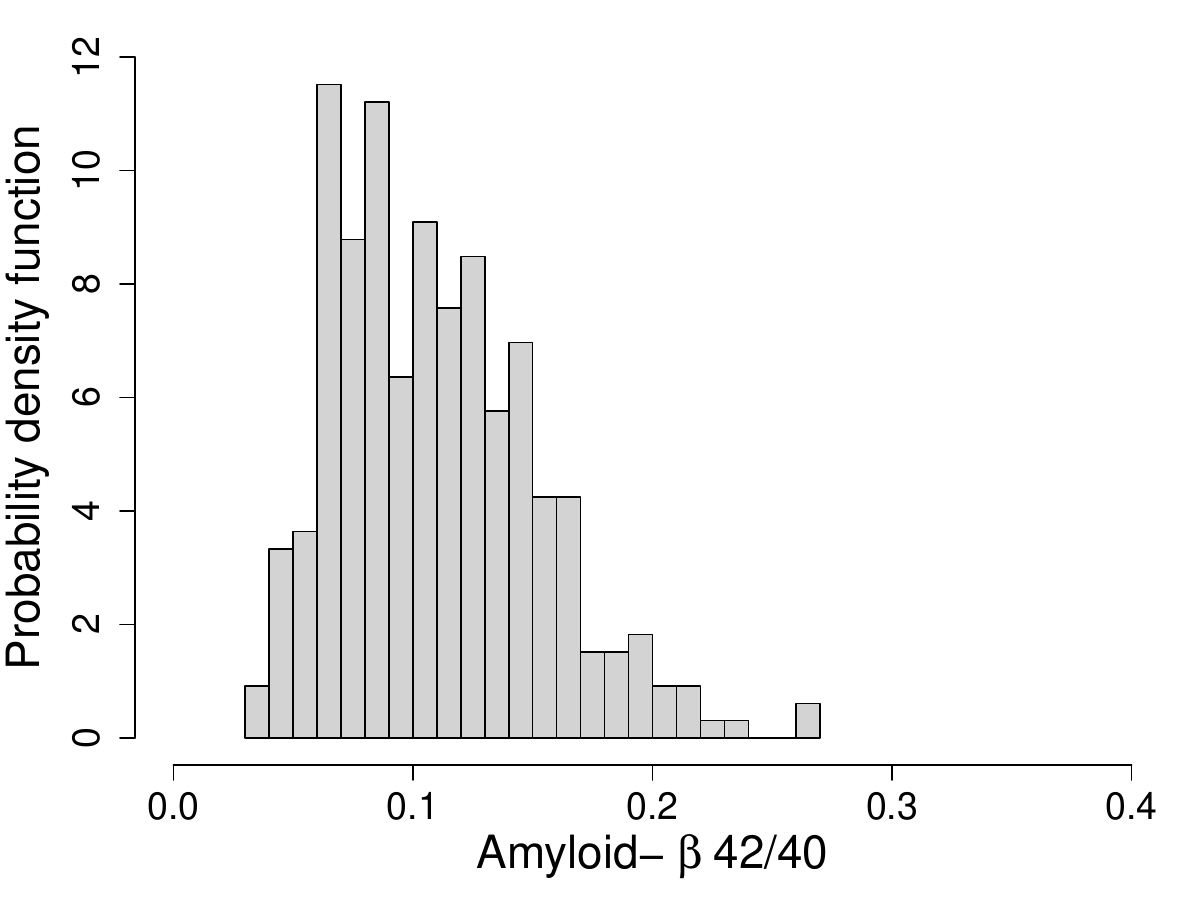}}
\subfloat[][]{\includegraphics[width=0.45\textwidth]{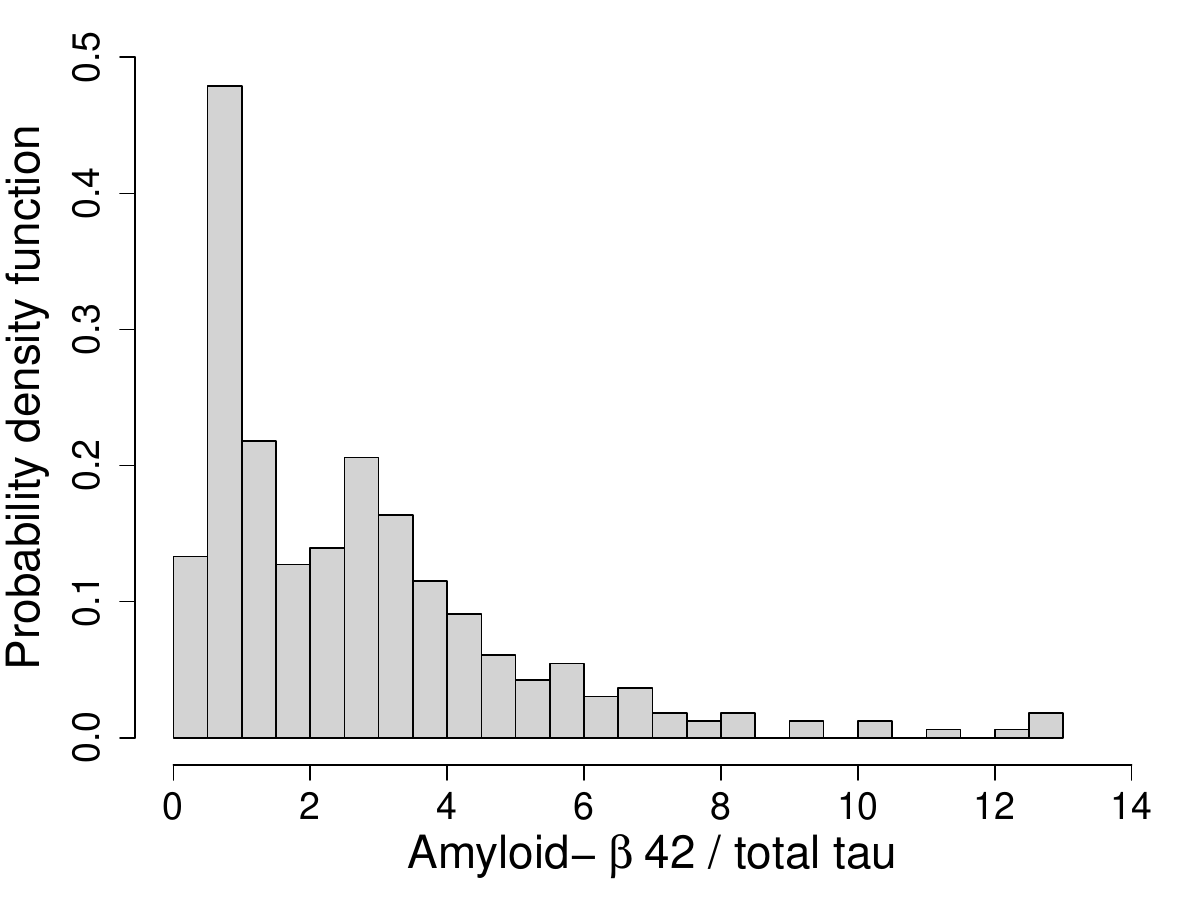}} \vspace{.2cm}
\caption{Analysis of the DCN study data. Distribution of the amyloid-$\beta$ 42/40 ratios (a) and the amyloid-$\beta$ 42/total tau ratios (b) in patients with MCI ($n=330$).}
\label{fig:hist_outcomes12}
\end{center}
\end{figure} 

\noindent {\it Description of the data.} In the DCN study, amyloid-$\beta$ and total tau baseline concentrations were measured in 374 patients diagnosed with MCI. In all other MCI patients, CSF biosamples were not collected due to either logistic reasons or lack of consent to the invasive procedure of lumbar puncture. Exclusion of patients that did not meet the inclusion criterion (age $\leq$ 50 years; $7$ patients) and of patients with missing values in at least one of the considered risk factors ($37$ patients) resulted in an analysis data set of $n=330$ patients. For details on the handling of missing values we refer to \citet{berger2019}. The unconditional distributions of the amyloid-$\beta$~42/40 ratio and the amyloid-$\beta$ 42/total tau ratio are visualized in Figure \ref{fig:hist_outcomes12}. While the values of the amyloid-$\beta$ 42/40 ratios are all smaller than 0.3, the amyloid-$\beta$ 42/total tau ratios range between 0.2 and 13, exhibiting a heavily right-skewed distribution. Kendall's rank correlation coefficient between the two components is measured to be \mbox{$\tau=0.307$} (amyloid-$\beta$ 42/40) and \mbox{$\tau=-0.269$} (amyloid-$\beta$ 42/total tau). Thus, our analysis had to deal with both positive and negative correlations between the ratio components. As mentioned before, this problem was our main motivation for the development of the FCGAM model.

The risk factors included in the analysis are summarized in Table \ref{tab:app_abeta_summary}. These were: (i) sex, (ii) age in years, (iii) educational level (measured by the number of years of education), and (iv) a binary variable indicating whether a patient was a carrier of the apolipoprotein E$\epsilon$4 (ApoE $\epsilon$4) allele, which is a strong genetic predictor of AD.\\

\begin{table}[!t]
\caption{Description and summary statistics of the two ratio outcomes and the covariates used for the analysis of the DCN study data (Q1 = first quartile, Q3 = third quartile). All numbers refer to a subset of patients diagnosed with MCI ($n = 330$). For details on the collection of the data, see \citet{kornhuber}.}
\begin{center}\footnotesize
\begin{tabular}{lrrrrrrr}
\toprule
Variable&\multicolumn{6}{c}{Summary statistics}\\
\midrule
&min & Q1 &median  & Q3&max& mean& sd\\[0.15cm]
Amyloid-$\beta$ 42/40 &0.03&0.08&0.10&0.14&0.26&0.11&0.04\\
Amyloid-$\beta$ 42/total tau &0.19&0.91&2.13&3.72&12.95&2.70&2.34\\[0.15cm]
Age (years) &51&60&66&73&89&66.51&8.11\\
Education (years) &2&11&11&13&19&12.18&2.96\\[0.15cm]
Sex &\multicolumn{1}{l}{\ male:}&\multicolumn{2}{r}{194 (58.8\%)}&&\multicolumn{1}{l}{\hspace{-0.3cm}female:}&\multicolumn{2}{r}{136 (41.2\%)}\\
ApoE $\epsilon$4 &\multicolumn{1}{l}{\ no:}&\multicolumn{2}{r}{182 (55.2\%)}&&\multicolumn{1}{l}{\hspace{-0.3cm}yes:}&\multicolumn{2}{r}{148 (44.8\%)}\\
\bottomrule
\end{tabular}
\end{center}
\label{tab:app_abeta_summary}
\end{table}

\noindent {\it Model fitting I.} In a preliminary analysis, we fitted GA and GA.LSS models for the components amyloid-$\beta$~42, amyloid-$\beta$ 40 and total tau, where either the rate parameters only (GA) or both the rate and the shape parameters (GA.LSS) were related to the four covariates. According to the Bayesian information criterion (BIC) the simple GA models (BIC $=4899.988$ for amyloid-$\beta$ 42, BIC $=6248.201$ for amyloid-$\beta$ 40 and\linebreak BIC $=4574.974$ for total tau) showed better fits than the respective GA.LSS models (BIC $=4914.396$ for amyloid-$\beta$ 42, BIC $=6266.375$ for amyloid-$\beta$ 40 and BIC $=4581.902$ for total tau). This result indicates that it is sufficient to relate the two rate parameters to the covariates. Furthermore, it supports the assumptions of the proposed FCGAM model, which treats the shape parameters $\delta_U$ and~$\delta_V$ as nuisance parameters. 

The fits of the FCGAM model with covariate-dependent association parameter are presented in Supplementary Table~S1. According to the credible intervals given in columns 4 and 6, none of the risk factors is found to affect the association parameter $\theta$. Applying Equation~\eqref{eq:tau} yielded the mean estimated rank correlations $\hat{\tau}(\hat{\theta}) = 0.35$ (range: 0.21 to 0.49) for amyloid-$\beta$ 42/40 and $\hat{\tau}(\hat{\theta}) = -0.23$ (range: -0.41 to 0.03) for amyloid-$\beta$ 42/total tau. Both estimates are  close to the respective unconditional rank correlations. \\

\noindent {\it Model fitting II.} Based on the above findings and to further reduce model complexity, we fitted FCGAM models with constant association parameter $\theta$ (setting the coefficients $\beta_{\theta,\text{Age}}, \ldots , \beta_{\theta,\text{ApoE $\epsilon4$}}$ to zero). We then calculated the BIC from these reduced models along with their counterparts obtained from the models with covariate-dependent $\theta$. For amyloid-$\beta$ 42/40, the BIC values were $11063.74$ (constant $\theta$) and~$11084.75$ (modeled $\theta$). For amyloid-$\beta$ 42/total tau, the BIC values were $9249.99$ (constant $\theta$) and $9267.485$ (modeled $\theta$). This result suggests that the reduced models with constant $\theta$ meet a better compromise between model fit and model complexity than the respective full models with covariate-dependent $\theta$. This observation is confirmed by randomized quantile residuals of the reduced FCGAM models (Figure~\ref{fig:qq_DCN}) indicating only slight deviations from normality. \\

\begin{figure}[!t]
\centering
\subfloat[][]{
\includegraphics[trim=0.5cm 0cm 0cm 0.5cm, width=0.35\textwidth]{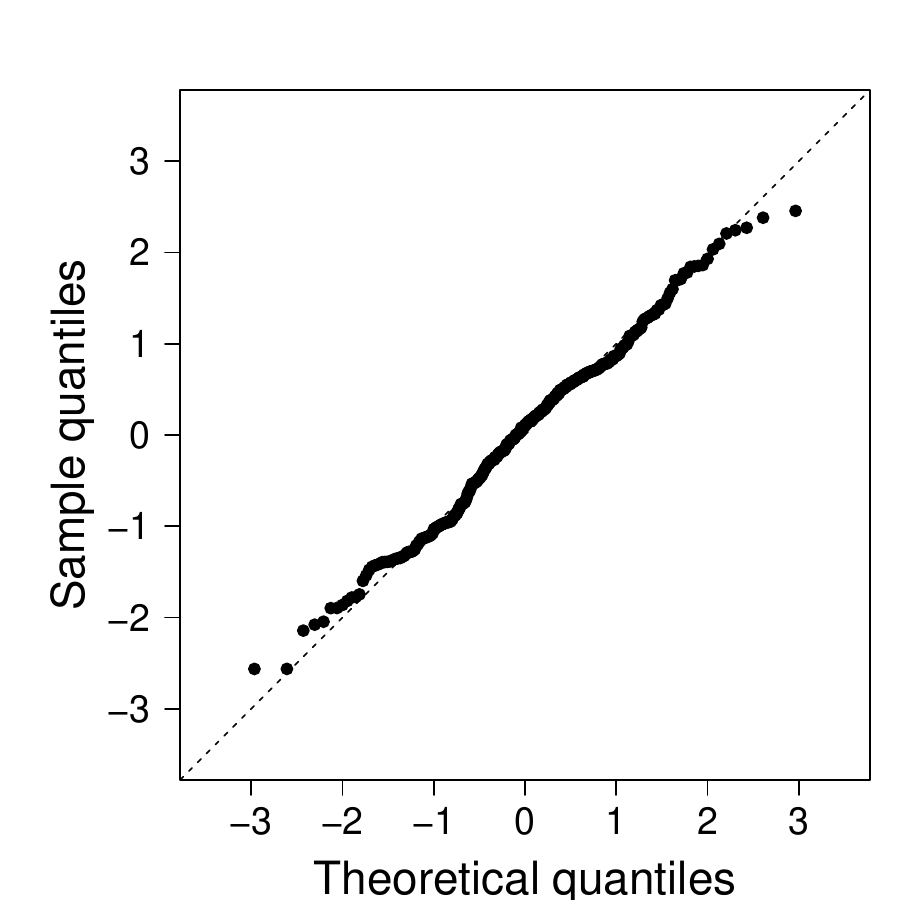}}
\subfloat[][]{
\includegraphics[trim=0.5cm 0cm 0cm 0.5cm, width=0.35\textwidth]{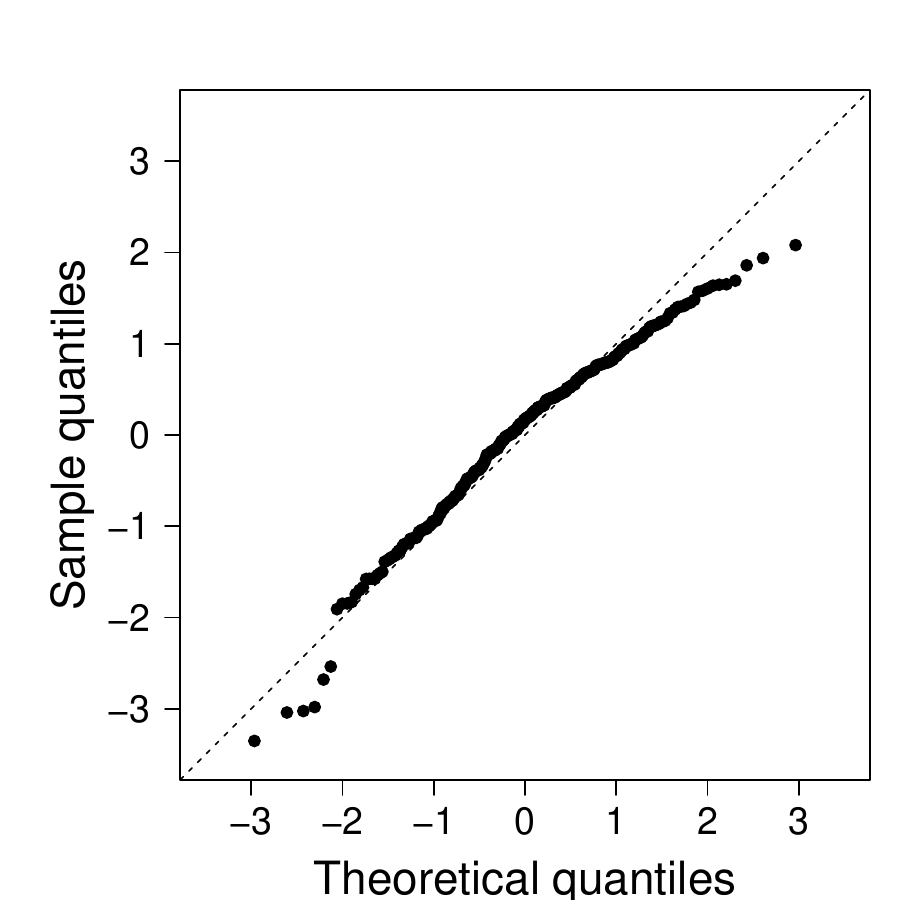}}  \vspace{.3cm}

\caption{Analysis of the amyloid-$\beta$ 42/40 ratios (a) and the amyloid-$\beta$ 42/total tau ratios (b) in the DCN study data. The figure presents normal quantile-quantile plots of the quantile residuals obtained from the FCGAM model fits.}
\label{fig:qq_DCN}
\end{figure}

\noindent {\it Main results.} The results obtained from the reduced FCGAM models are shown in Table~\ref{tab:res_DCN2} and Supplementary Figures~S5 and~S6. The upper part of Table~\ref{tab:res_DCN2} refers to the parameter $\Lambda=\lambda_U/\lambda_V$, reporting the differences $\hat{\beta}_{\Lambda j} := \hat{\beta}_{Uj}-\hat{\beta}_{Vj}$. Note that the coefficient estimates are very similar to the respective estimates of the more complex model in Table~S1. For example, for amyloid-$\beta$ 42/total tau one obtains $\hat{\beta}_{\Lambda , \text{ApoE $\epsilon 4$}}=0.3786$ (Table~\ref{tab:res_DCN2}) and $\hat{\beta}_{\Lambda , \text{ApoE $\epsilon 4$}}=0.2411 + 0.1406 = 0.3817$ (Table~S1). The credible intervals in Table~\ref{tab:res_DCN2} were obtained by drawing a sample of size 10,000 from the posterior distribution in \eqref{eq:posterior} and by calculating the $2.5\%$ and $97.5\%$ percentiles from the sampled differences ${\beta}_{Uj}-{\beta}_{Vj}$. According to the results of the FCGAM model, there is strong evidence for an effect of the risk factors age and ApoE $\epsilon$4 on the\linebreak amyloid-$\beta$ 42/40 and amyloid-$\beta$ 42/total tau ratios. As depicted in Supplementary Figures~S5(a) and~S6(a), both the expected amyloid-$\beta$ 42/40 ratio and the expected amyloid-$\beta$ 42/total tau ratio decrease with increasing age, implying a higher risk of progression to AD in older patients. Similarly, the expected ratios of ApoE~$\epsilon$4 carriers are strongly reduced compared to patients not carrying the allele (Supplementary Figures~S5(d) and~S6(d), confirming the important role of this genetic risk factor in AD progression). The figures also illustrate how the estimated median values as well as the $10\%$ and $90\%$ percentiles of the distributions change with the covariates. In contrast to age and ApoE~$\epsilon$4, Table~\ref{tab:res_DCN2} shows no evidence for an effect of sex and educational level on the two ratio outcomes. These results are in full agreement with the findings by~\citet{berger2019}, who fitted an eGB2 model with amyloid-$\beta$ 42/40 outcome to the DCN study data.

\begin{table}[!t]
\caption{Analysis of the amyloid-$\beta$ 42/40 ratios (left) and the amyloid-$\beta$ 42/total tau ratios (right) in the DCN study data. The table presents the coefficient estimates with $95\%$ credible intervals (calculated by the procedure described in Section~\ref{sec:fe_ci}), as obtained from fitting FCGAM models with constant association parameter~$\theta$.}
\begin{center}\footnotesize
\begin{tabular}{llrrrr}
\toprule
&&\multicolumn{2}{c}{\textbf{amyloid-$\boldsymbol{\beta}$ 42/40}}&\multicolumn{2}{c}{\textbf{amyloid-$\boldsymbol{\beta}$ 42/total tau}}\\[.1cm]
Parameter&Covariate&$\hat{\beta}$& $95\%$ CI&$\hat{\beta}$& $95\%$ CI\\[.05cm]
\midrule
$\Lambda$&Age (years)      & $0.0089$    & $[0.0041;0.0137]$ & $0.0256$    & $[0.0156;0.0377]$ \\
&Education (years) & $-0.0017$    & $[-0.0152;0.0117]$ & $-0.0150$    & $[-0.0461;0.0168]$\\
&Sex (male) & . & . & . & .\\
&Sex (female)    & $0.0724$    & $[-0.0073;0.1547]$ & $0.0283$    & $[-0.1544;0.2146]$\\
&ApoE $\epsilon$4 (no) & . & . & . & .\\
&ApoE $\epsilon$4 (yes) & $0.1967$  & $[0.1157;0.2766]$ & $0.3786$  & $[0.1965;0.5584]$\\[.05cm]
\midrule \\[-.3cm]
$\theta$&& $3.5325$ & $[2.7671;4.2888]$ & $-2.0611$ & $[-2.8064;-1.3073]$\\
$\delta_U$&& $6.0586$ & $[5.1169;7.0218]$ & $5.8090$ & $[4.9062;6.7092]$\\
$\delta_V$&& $10.0151$ & $[8.5182;11.5039]$ & $2.6718$ & $[2.2900;3.0542]$\\
\bottomrule
\end{tabular}
\end{center}
\label{tab:res_DCN2}
\end{table}

\section{Discussion}\label{sec:discussion}

The main contribution of this work is a copula-based regression model that relates the ratio of two gamma distributed components to a set of covariates. Our model is primarily designed for the analysis of ratio outcomes in medical research, which is an important task, for instance, in neurology \citep{novellino}, infectiology \citep{caby2016} and pharmacology \citep{cawley}. Importantly, when biomarker ratios are used as clinical metrics or indicators of clinical outcomes, our model may be used to relate the respective ratio values to a set of risk factors and/or confounding variables.  A prototypical example is given by the prognosis of AD progression considering ratios of amyloid-$\beta$ and total tau protein biomarkers, as presented in Section~\ref{sec:app} of this paper.

Conceptually, the FCGAM model developed in this paper has the following advantages: First, by assuming the ratio components to follow marginal gamma distributions, the FCGAM model represents the two biomarkers by real-valued random variables with positive support and right-skewed (marginal) distributions. These distributional characteristics, which are common to most biomarkers encountered in medical research, are directly incorporated in the definition of the proposed copula model. As a consequence, the resulting ratio density incorporates the full information contained in the marginal densities of the components of the ratio. We emphasize that this property does not apply to simpler modeling approaches approximating the ratio by a single log-normal or gamma-distributed variable. In fact, without consideration of the paired components themselves, these approximations inevitably bear the risk of a loss of information \citep[``neglected companion'';][]{kerkhof}. This issue has been demonstrated by the results of our simulation study (Section \ref{sec:sim}), which
resulted in an increased estimation accuracy of the proposed copula-based approach in all data-generating scenarios. We also stress that linking the two marginal distributions by a copula does in general not restrict our model to the use of two gamma distributions for the ratio components. In fact, although our model can be seen as the most relevant use case in many medical applications, the marginal distributions can in principle be replaced by arbitrary parametric distributions. For instance, our model can in a straightforward manner be extended to situations where one biomarker is discrete or ordinal. 

Second, the proposed FCGAM model has a high flexibility regarding the direction of the association between the two ratio components. Importantly, by the choice of Frank copula, the FCGAM model allows for both positive and negative values of the (rank) correlation between the components $U$ and $V$, thereby improving previous modeling approaches that restricted this correlation to be zero \citep{YeeVGAMBook} or positive \citep{berger2019}. As demonstrated in the simulation study in Section \ref{sec:sim}, the FCGAM model indeed performs better in terms of estimation accuracy when the association between $U$ and $V$ is negative. On the other hand, it does not perform worse than the aforementioned approach when the association between $U$ and $V$ is positive.

Third, although the proposed model incorporates the full information contained in the marginal densities $f_U$ and $f_V$, it provides a rather simple interpretation of the associations between the ratio $U/V$ and the covariates. This is because the FCGAM model reduces the original five-parameter set $(\lambda_U, \delta_U, \lambda_V, \delta_V, \theta)^\top$ (including all parameters of the marginal densities and the association parameter $\theta$) to the restricted set $(\Lambda , \delta_U, \delta_V, \theta)^\top$ with $\Lambda = \lambda_U / \lambda_V$. As a consequence, when treating $\delta_U$, $\delta_V$ (and possibly also $\theta$) as nuisance parameters, the associations between $U/V$ and each of the covariates can be investigated using one-dimensional coefficient estimates and single-parameter hypothesis tests. Similarly, the association between the components $U$ and $V$ has a natural interpretation in terms of Kendall's rank correlation, being related to $\theta$ by the one-to-one relationship given in Equation \eqref{eq:tau}.

Finally, beside the flexibility in specifying other marginal distributions than the gamma distribution, the FCGAM model may be extended in many other ways. For example,  Frank copula could be replaced by other copulas \citep[noting that the results on ratio densities are also valid for other absolutely continuous copulas; see][]{ly2019}. When there is particular interest in the tail dependencies of $U$ and $V$, benchmark experiments to identify the best fitting copula and/or marginal distributions could be performed using resampling techniques (e.g.\@ bootstrapping or subsampling). It should be noted, however, that other copulas from the literature might be less flexible regarding the range of $\theta$ \citep[and thus also the range of possible associations between the components $U$ and $V$; see e.g.][for a recent overview of copulas allowing for modeling negative dependence]{GhoBhuFin2022}. For example, it is not possible to model negative associations between $U$ and $V$ using non-rotated Gumbel or Joe copulas.

Despite our biostatistical focus, the proposed FCGAM methodology is a general statistical modeling approach that can in principle be used in any research discipline. Interesting areas are e.g.\@ environmental research \citep{perri} and information engineering \citep{mekic2012}.

\section{Appendices}

It contains the proof of Proposition 2 (Appendix A), additional simulation results (Appendix B) as well as further results of the analysis of the DCN study data (Appendix C). 

\section{Acknowledgments}
Moritz Berger acknowledges support by the grant BE 7543/1-1 of the German research foundation (DFG). Nadja Klein acknowledges support by the Emmy Noether grant KL 3037/1-1 of the DFG.  The analysis of the DCN study data was supported by the German Federal Ministry of Education and Research (Kompetenznetz Demenzen, grant 01GI0420)

\bibliographystyle{plainnat}
\bibliography{bib}

\appendix

\renewcommand{\thesection}{Appendix~\Alph{section}}
\renewcommand\thefigure{\Alph{figure}}   
\renewcommand\thetable{\Alph{table}} 
\setcounter{figure}{0}
\setcounter{table}{0}
\section{Proof of Proposition 2}\label{app:A}

By Proposition 1, the PDF of the ratio $R=U/V$ is given by 
\begin{align*}
f_R(r;\lambda_U,\lambda_V,\delta_U,\delta_V,\theta)=\int_{0}^{1}&\,c_\theta\left(F_U(r\,F_V^{-1}(s; \lambda_V, \delta_V); \lambda_U, \delta_U),s\right)\\ 
& \times \, \left|F_V^{-1}(s; \lambda_V, \delta_V)\right| \, f_U(rF_V^{-1}(s; \lambda_V, \delta_V), \lambda_U, \delta_U)\,ds\\
=\int_{0}^{1}&c_\theta\left(\frac{1}{\Gamma(\delta_U)} \, \gamma\left(\delta_U,\lambda_U\,r\,F_V^{-1}(s; \lambda_V, \delta_V)\right),s\right)\\
& \times \, F_V^{-1}(s; \lambda_V, \delta_V) \ \frac{\lambda_U^{\delta_U}}{\Gamma(\delta_U)}\left(r\,F_V^{-1}(s; \lambda_V, \delta_V)\right)^{\delta_U-1}\\
& \times \, \exp\left(-\lambda_U\,r\,F_V^{-1}(s; \lambda_V, \delta_V)\right)\,ds\\
=\int_{0}^{1}&c_\theta\left(\frac{1}{\Gamma(\delta_U)} \, \gamma\left(\delta_U,\lambda_U\,r\ \frac{\gamma^{-1}\left(\delta_V,\Gamma(\delta_V)\,s\right)}{\lambda_V}\right),s\right)\\
& \times \, \frac{\lambda_U^{\delta_U}}{\Gamma(\delta_U)}\left(\frac{\gamma^{-1}\left(\delta_V,\Gamma(\delta_V)\,s\right)}{\lambda_V}\right)^{\delta_U}r^{\delta_U-1}\\
& \times \, \exp\left(-\lambda_U\,r\ \frac{\gamma^{-1}\left(\delta_V,\Gamma(\delta_V)\,s\right)}{\lambda_V}\right)\,ds\\
\overset{\Lambda = \lambda_U / \lambda_V}{=}\int_{0}^{1}&{c_\theta\left(\frac{1}{\Gamma(\delta_U)}\,\gamma\left(\delta_U, r\,\Lambda\, \gamma^{-1}\left(\delta_V, \Gamma(\delta_V)s\right)\right),s\right)} \nonumber \\
& \times \, \frac{\Lambda^{\delta_U}\, r^{(\delta_U-1)}}{\Gamma(\delta_U)}\, \left(\gamma^{-1}\left(\delta_V, \Gamma(\delta_V)s\right)\right)^{\delta_U} \nonumber \\
& \times \, \exp\left(-r\, \Lambda\, \gamma^{-1}\left(\delta_V, \Gamma(\delta_V)s\right)\right)\,ds\\
= f_R & (r; \Lambda, \delta_U, \delta_V, \theta)\, ,
\end{align*}
where $\gamma(\cdot,\cdot)$ denotes the lower incomplete gamma function. \QEDA

\newpage

\section{Further Simulation Results}

\renewcommand{\thefigure}{S\arabic{figure}}
\renewcommand{\thetable}{S\arabic{table}}

\begin{figure}[H]
\begin{center}
\includegraphics[width=0.8\textwidth]{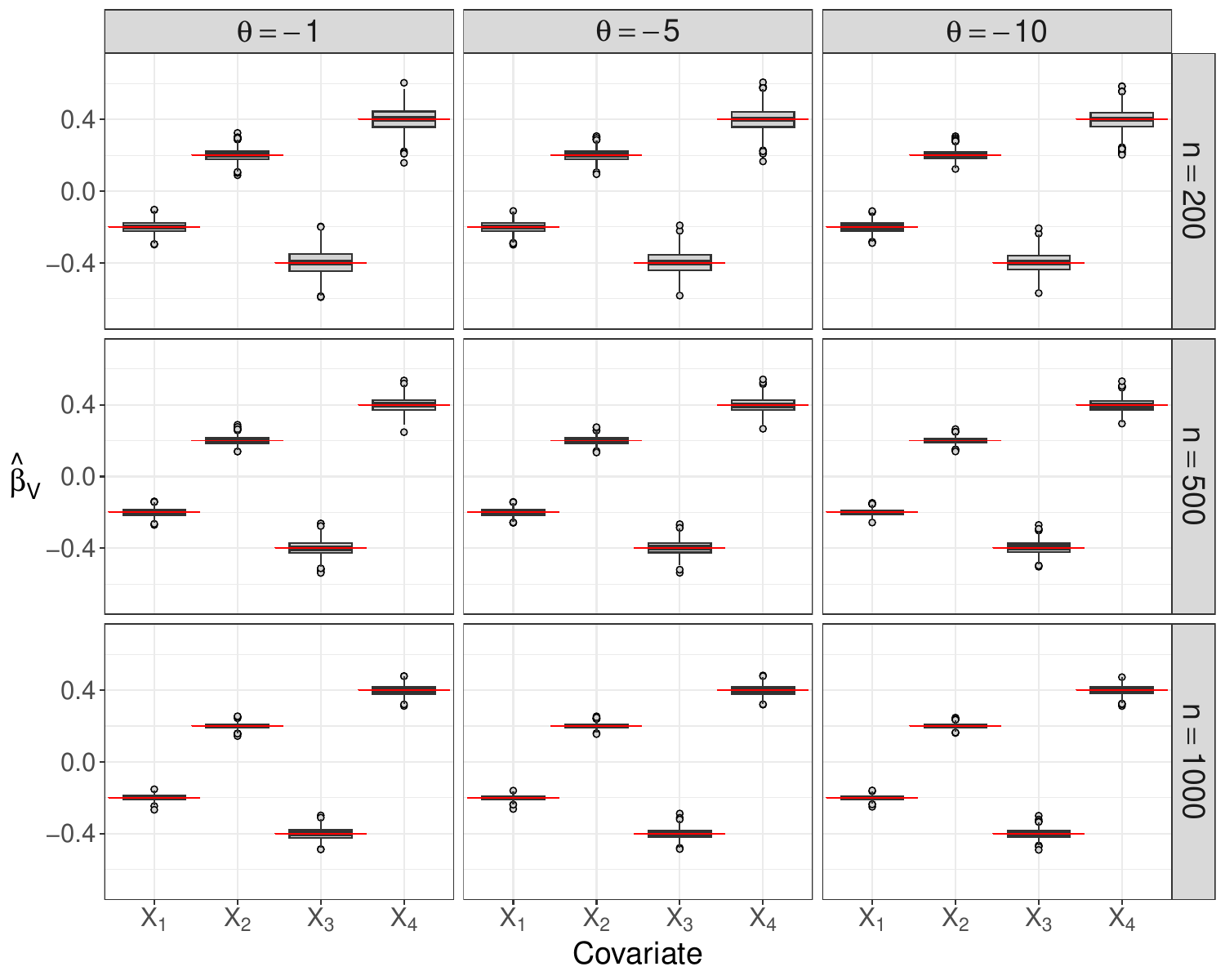}
\caption{Point estimates of the FCGAM coefficients in \textit{Simulation Study 1}. The boxplots visualize the MLEs of the coefficients $\beta_{V1}=-0.2$, $\beta_{V2}=0.2$, $\beta_{V3}=-0.4$ and $\beta_{V4}=0.4$ that were obtained from fitting the FCGAM model to 1000 data sets of size $n$ each. The red lines refer to the true values of the coefficients.}
\label{fig:sim1:esty}
\end{center}
\end{figure} 

\begin{figure}[H]
\begin{center}
\includegraphics[width=0.8\textwidth]{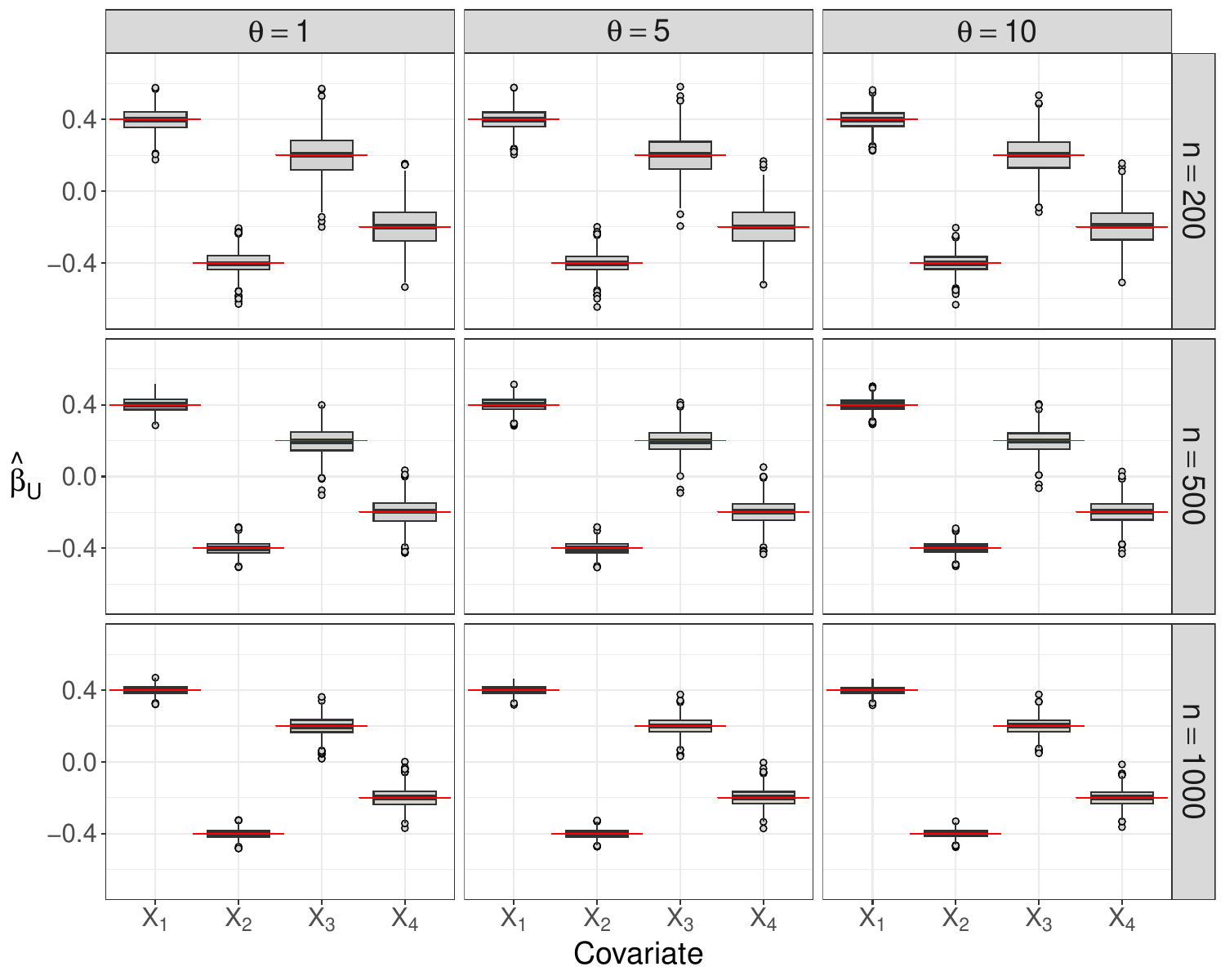}
\caption{Point estimates of the FCGAM coefficients in \textit{Simulation Study 2}. The boxplots visualize the MLEs of the coefficients $\beta_{U1}=0.4$, $\beta_{U2}=-0.4$, $\beta_{U3}=0.2$ and $\beta_{U4}=-0.2$ that were obtained from fitting the FCGAM model to 1000 data sets of size $n$ each. The red lines refer to the true values of the coefficients.}
\label{fig:sim2:estx}
\end{center}
\end{figure} 

\begin{figure}[H]
\begin{center}
\includegraphics[width=0.8\textwidth]{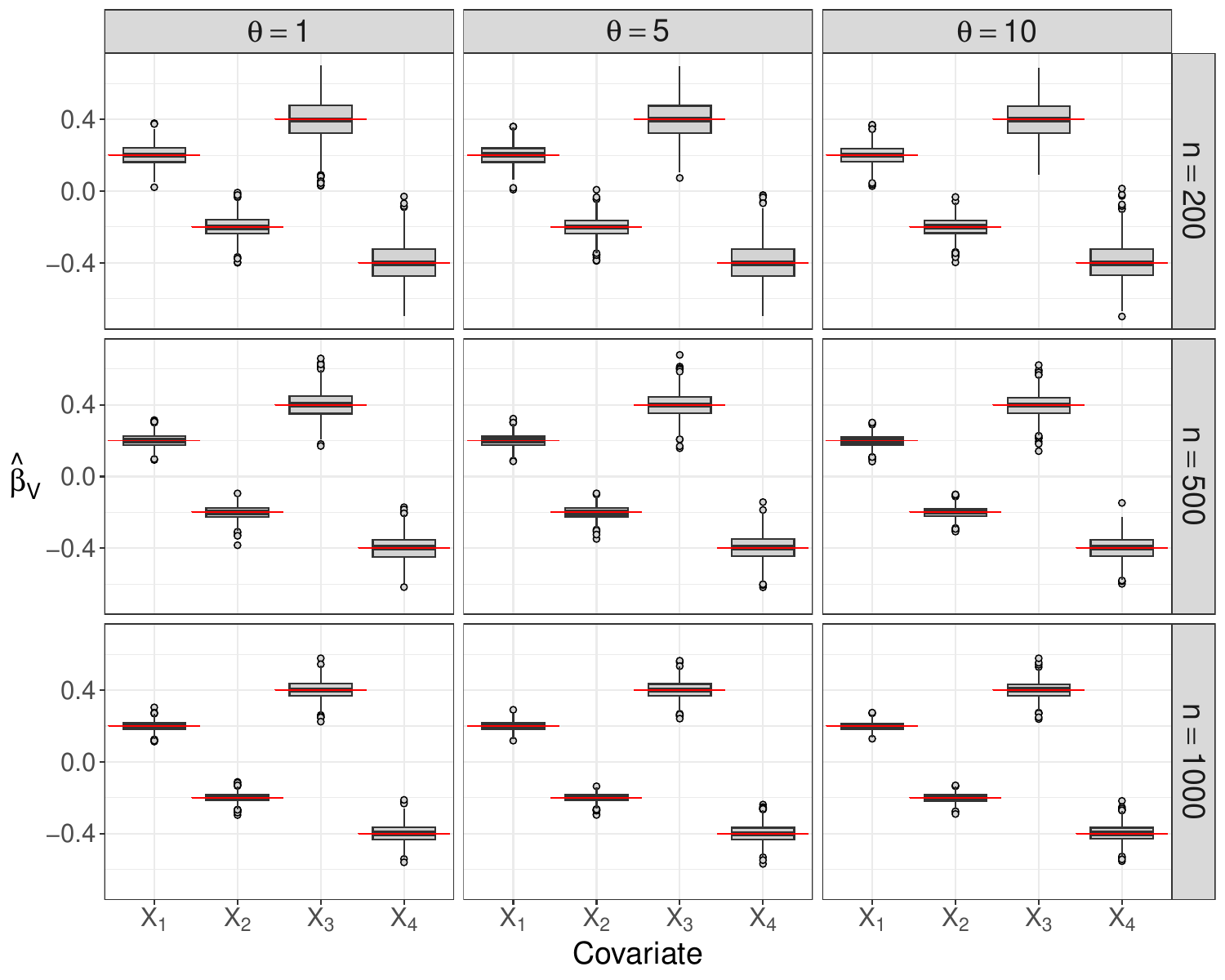}
\caption{Point estimates of the FCGAM coefficients in \textit{Simulation Study 2}. The boxplots visualize the MLEs of the coefficients $\beta_{V1}=0.2$, $\beta_{V2}=-0.2$, $\beta_{V3}=0.4$ and $\beta_{V4}=-0.4$ that were obtained from fitting the FCGAM model to 1000 data sets of size $n$ each. The red lines refer to the true values of the coefficients.
}
\label{fig:sim2:esty}
\end{center}
\end{figure}

\begin{figure}[H]
\begin{center}
\includegraphics[width=0.8\textwidth]{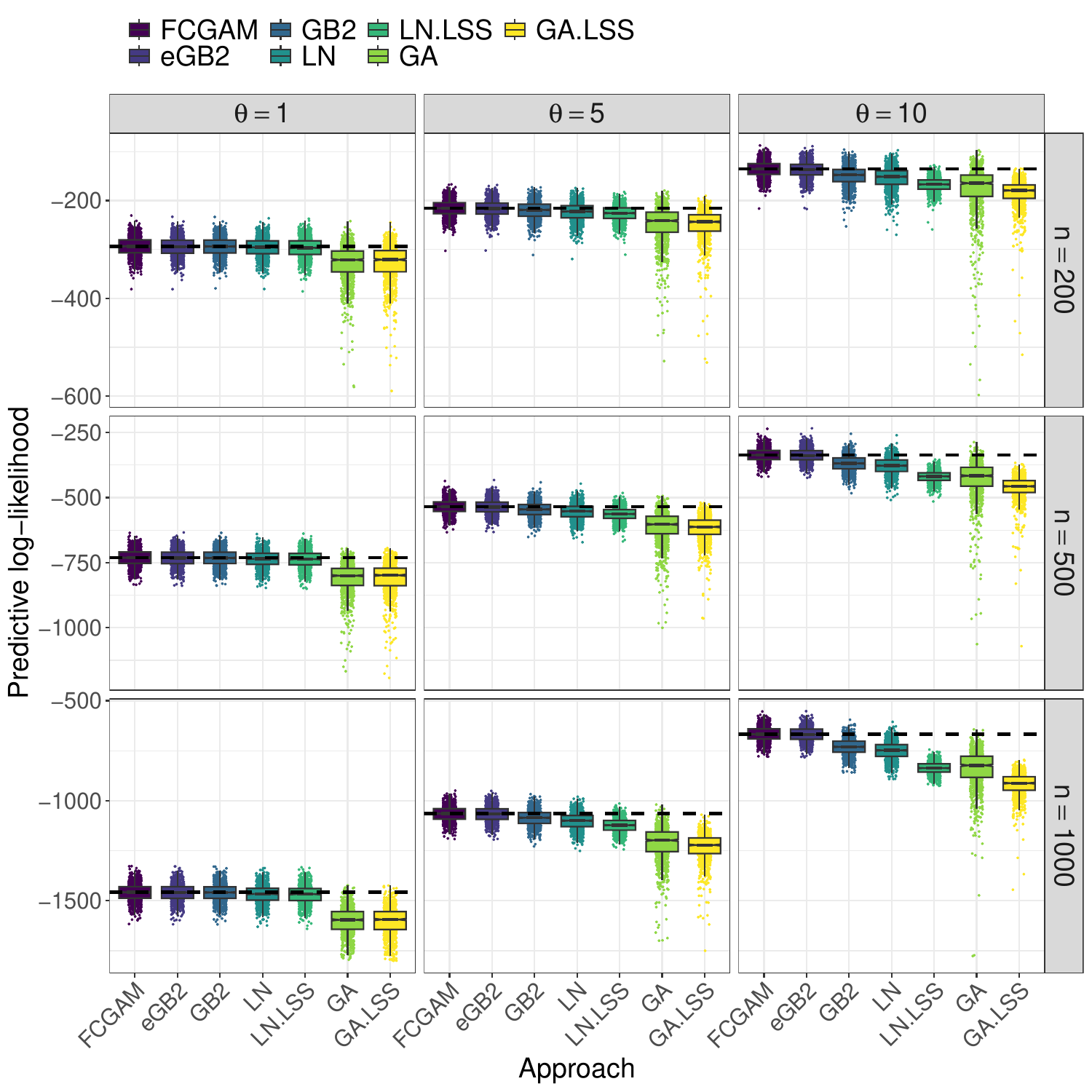}
\caption{Comparison of the FCGAM model to alternative methods in \textit{Simulation Study 2}. The boxplots visualize the predictive log-likelihood values obtained from the FCGAM model and from the benchmark methods (ii) to (vii). All models were fitted to 1000 independent data sets and evaluated on independently generated test data sets of the same size. In each panel, the dashed horizontal line indicates the median predictive log-likelihood value of the best performing method.}
\label{fig:sim2:ll}
\end{center}
\end{figure} 

\newpage

\section{Further Results of the Analysis of the DCN Study Data}

\begin{table}[H]
\caption{Analysis of the amyloid-$\beta$ 42/40 ratios (left) and the amyloid-$\beta$ 42/total tau ratios (right) in the DCN study data. The table presents the coefficient estimates with $95\%$ credible intervals (calculated by the procedure described in Section~2.3), as obtained from fitting FCGAM models with covariate-dependent association parameter~$\theta$.}
\begin{center}\footnotesize
\begin{tabular}{llrrrr}
\toprule
&&\multicolumn{2}{c}{\textbf{amyloid-$\boldsymbol{\beta}$ 42/40}}&\multicolumn{2}{c}{\textbf{amyloid-$\boldsymbol{\beta}$ 42/total tau}}\\[.1cm]
Parameter&Covariate&$\hat{\beta}$& $95\%$ CI&$\hat{\beta}$& $95\%$ CI\\[.1cm] 
\midrule 
$\lambda_U$&Age (years)       &$0.0073$    &$[0.0019;0.0127]$&$0.0092$    &$[0.0037;0.0148]$\\
&Education (years) & $0.0079$    &$[-0.0071;0.0224]$& $0.0035$    &$[-0.0123;0.0192]$\\
&Sex (male) & . & . & . & .\\
&Sex (female)    &$0.0017$    &$[-0.0905;0.0933]$&$-0.0247$    &$[-0.1178;0.0680]$\\
&ApoE $\epsilon$4 (no) & . & . & . & .\\
&ApoE $\epsilon$4 (yes) &$0.1774$  &$[0.0871;0.2681]$&$0.2411$  &$[0.1503;0.3325]$\\[.05cm]
\midrule \\[-.3cm]
$\lambda_V$&Age (years)      &$-0.0015$    &$[-0.0056;0.0025]$&$-0.0160$    &$[-0.0242;-0.0079]$\\
&Education (years) & $0.0100$    &$[-0.0013;0.0212]$& $0.0166$    &$[-0.0071;0.0406]$\\
&Sex (male) & . & . & . & .\\
&Sex (female)    &$-0.0751$    &$[-0.1460;-0.0061]$&$-0.0746$    &$[-0.2127;0.0663]$\\
&ApoE $\epsilon$4 (no) & . & . & . & .\\
&ApoE $\epsilon$4 (yes) &$-0.0141$  &$[-0.0811;0.0526]$&$-0.1406$  &$[-0.2745;-0.0065]$\\[.05cm]
\midrule \\[-.3cm]
$\theta$&Age (years)       &$0.0369$    &$[-0.0645;0.1347]$&$0.0637$    &$[-0.0356;0.1640]$\\
&Education (years) & $0.1247$    &$[-0.1453;0.3928]$& $0.0165$    &$[-0.2416;0.2750]$\\
&Sex (male) & . & . & . & .\\
&Sex (female)    &$-0.5679$    &$[-2.1432;1.0406]$&$-0.9324$    &$[-2.4944;0.6044]$\\
&ApoE $\epsilon$4 (no) & . & . & . & .\\
&ApoE $\epsilon$4 (yes) &$0.5610$  &$[-0.9635;2.0978]$& $-1.2897$  &$[-2.8033;0.1928]$\\[.05cm]
\midrule \\[-.3cm]
$\delta_U$&&$5.7822$&$[4.9088;5.6759]$&$5.6900$&$[4.8084;6.5605]$\\
$\delta_V$&&$10.0740$&$[8.4946;11.6237]$&$2.6023$&$[2.2232;2.9831]$\\
\bottomrule
\end{tabular}
\end{center}
\label{tab:res_DCN1}
\end{table}

\begin{figure}[p]
\begin{center}
\subfloat[][]{\includegraphics[width=0.45\textwidth]{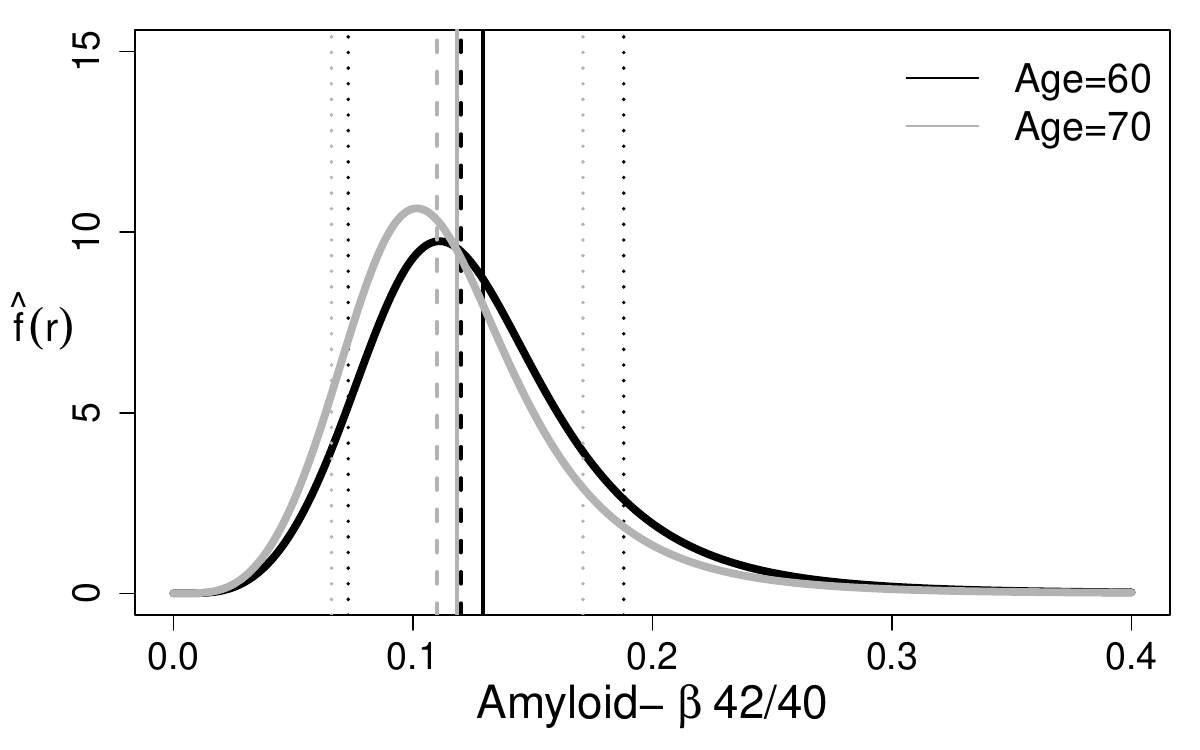}}
\hspace{0.05cm}
\subfloat[][]{\includegraphics[width=0.45\textwidth]{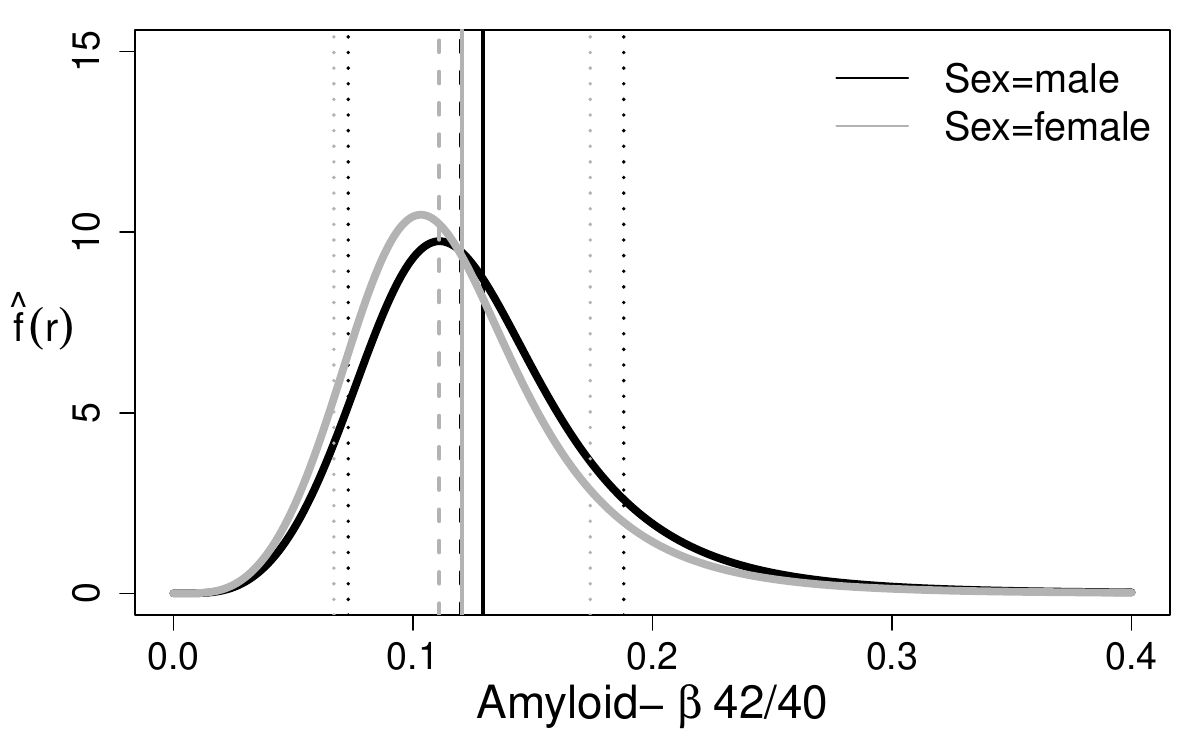}}

\subfloat[][]{\includegraphics[width=0.45\textwidth]{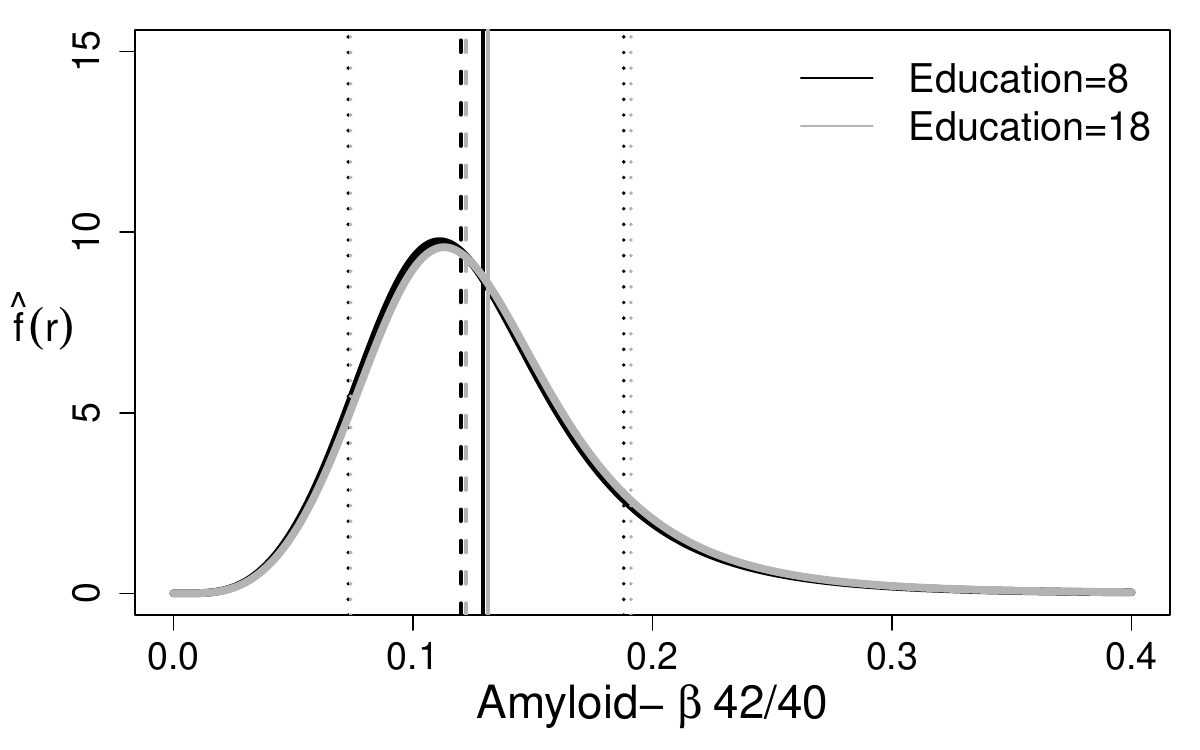}}
\hspace{0.05cm}
\subfloat[][]{\includegraphics[width=0.45\textwidth]{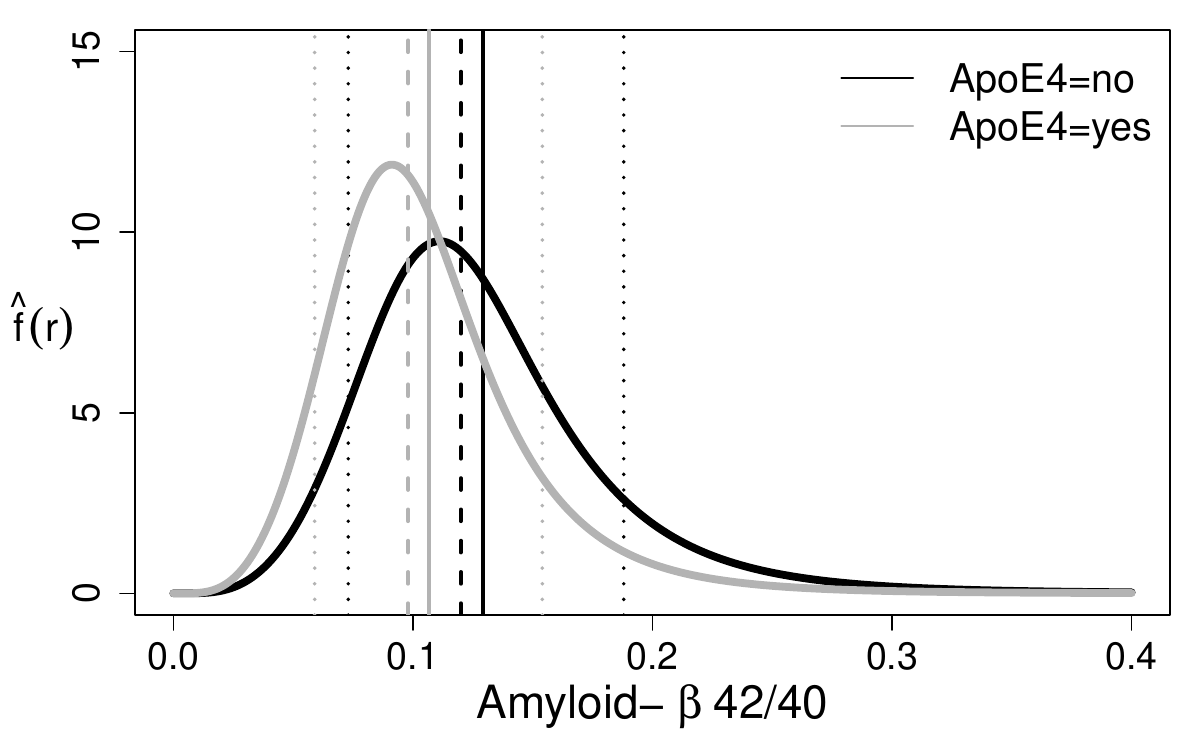}} \vspace{.3cm}
\caption{Analysis of the amyloid-$\beta$ 42/40 ratios in the DCN study data. The black lines refer to the estimated PDFs for a covariate profile of a randomly selected study participant (60 years of age, Sex = male, Education = 8 years, ApoE $\epsilon$4 = no). The gray lines refer to a situation where the participant would have been 70 years of age (a), would have been female (b), would have had 18 years of education (c), and would have been a carrier of the ApoE~$\epsilon$4 allele (d). The vertical lines correspond to the estimated mean values (solid), median values (dashed) and $10\%$ and $90\%$ percentiles (dotted).} 
\label{fig:res_R4240}
\end{center}
\end{figure} 

\begin{figure}[p]
\begin{center}
\subfloat[][]{\includegraphics[width=0.45\textwidth]{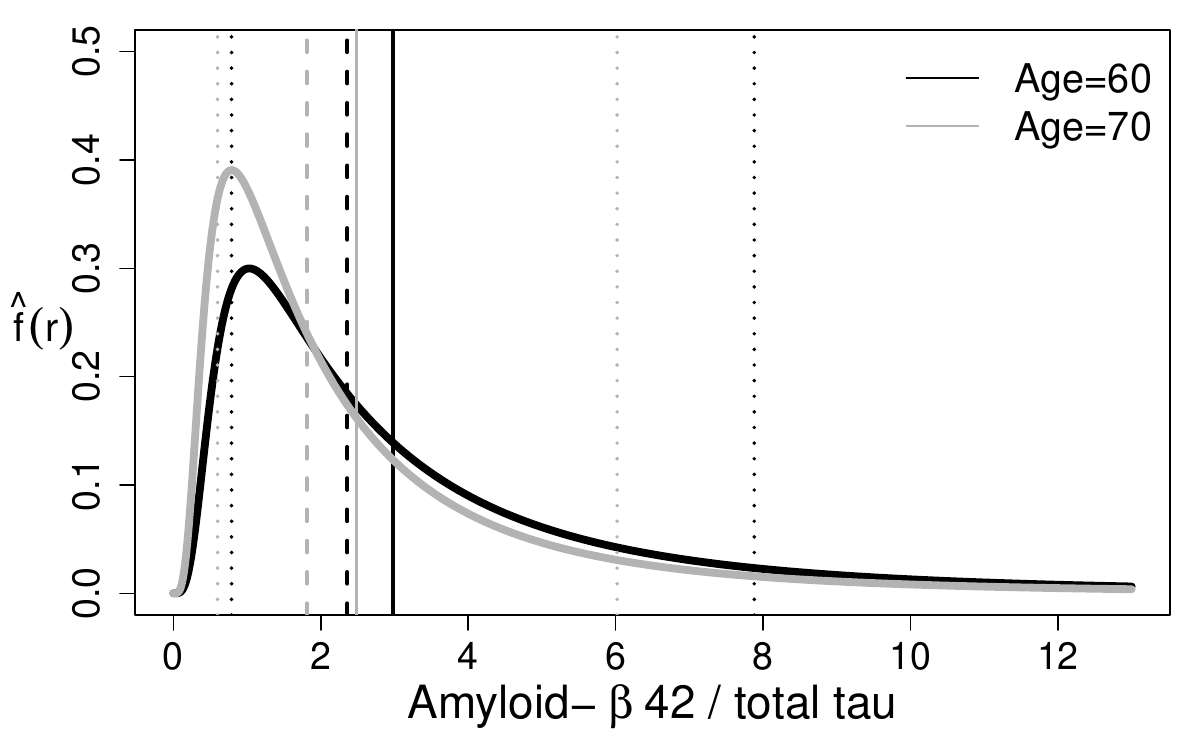}}
\hspace{0.05cm}
\subfloat[][]{\includegraphics[width=0.45\textwidth]{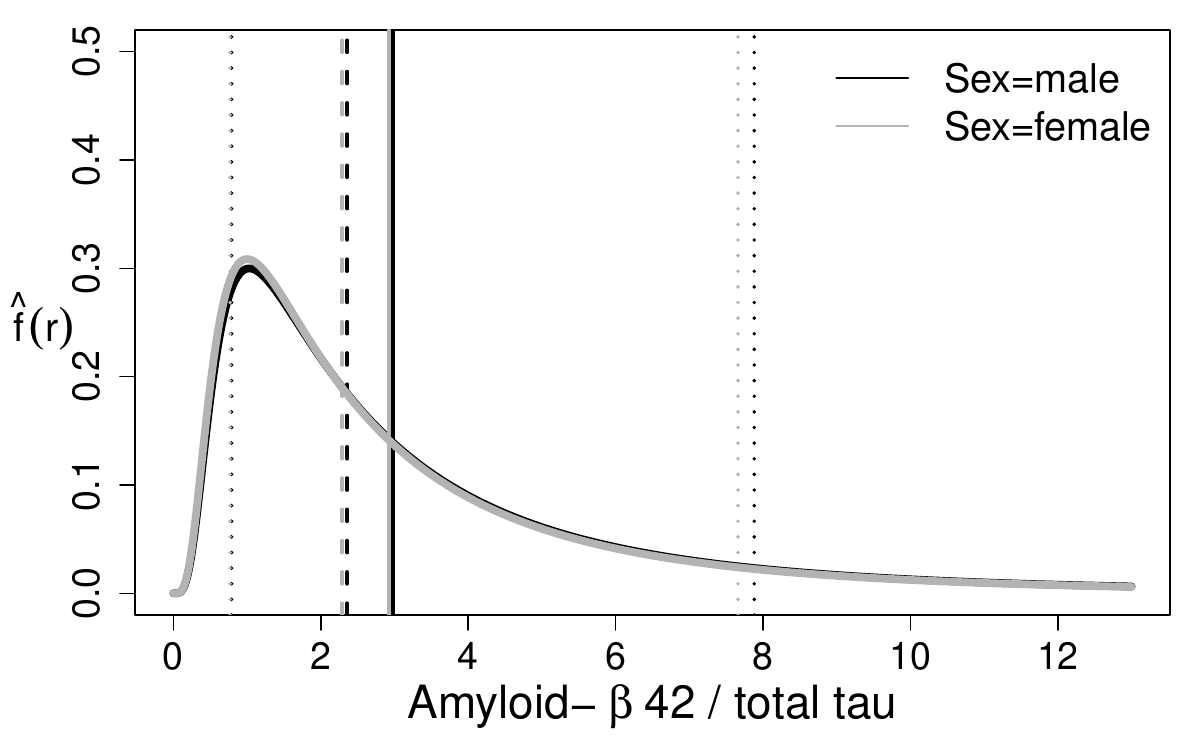}}

\subfloat[][]{\includegraphics[width=0.45\textwidth]{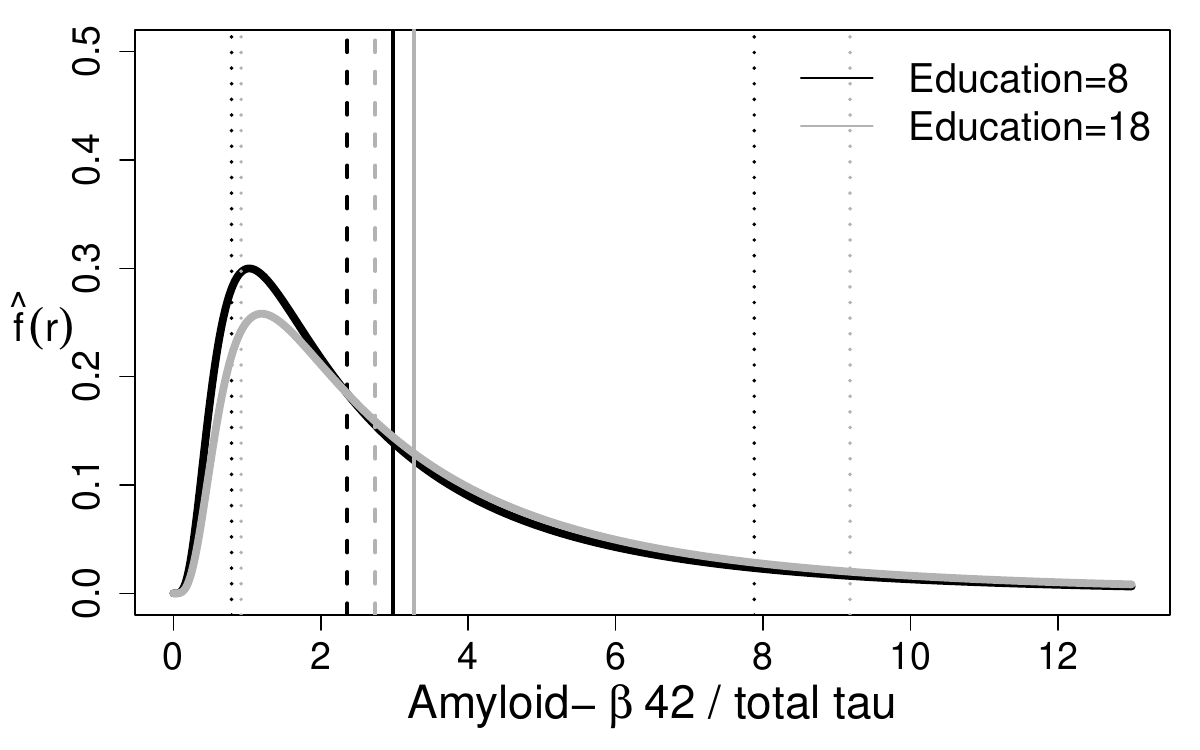}}
\hspace{0.05cm}
\subfloat[][]{\includegraphics[width=0.45\textwidth]{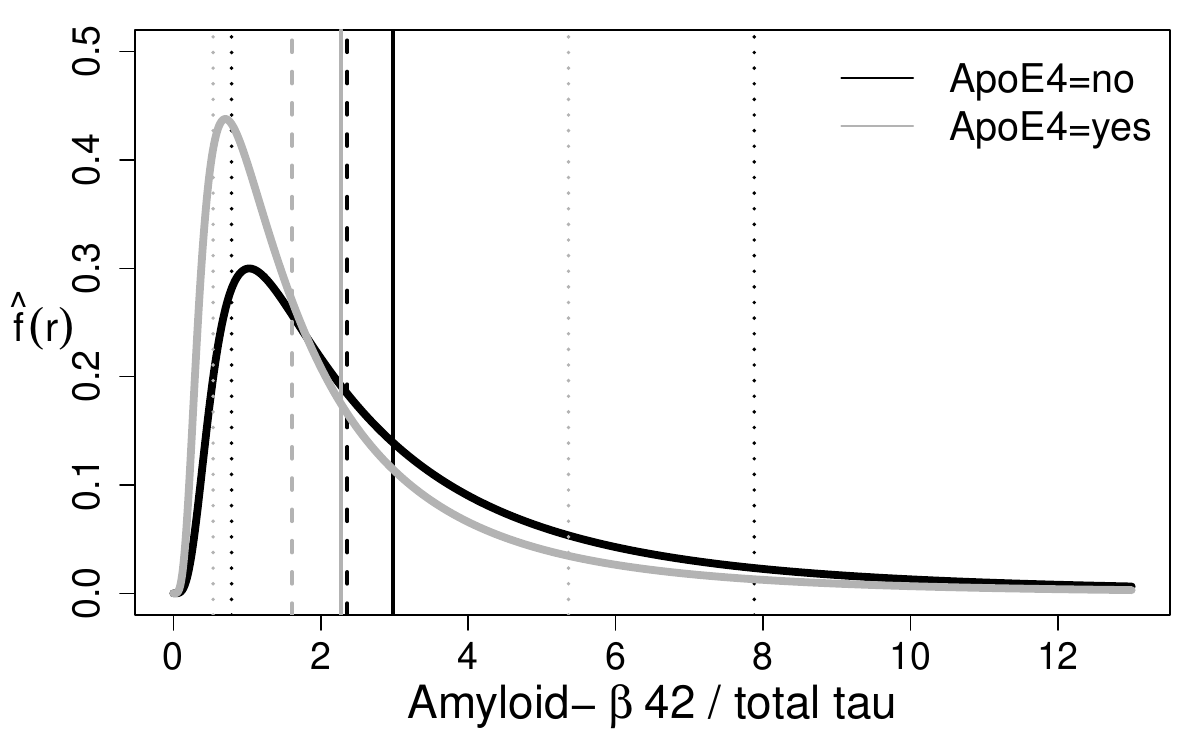}}\vspace{.3cm}
\caption{Analysis of the amyloid-$\beta$ 42/total tau ratios in the DCN study data. The black lines refer to the estimated PDFs for a covariate profile of a randomly selected study participant (60 years of age, Sex = male, Education = 8 years, ApoE $\epsilon$4 = no). The gray lines refer to a situation where the participant would have been 70 years of age (a), would have been female~(b), would have had 18 years of education (c), and would have been a carrier of the ApoE~$\epsilon$4 allele (d). The vertical lines correspond to the estimated mean values (solid), median values (dashed) and $10\%$ and $90\%$ percentiles (dotted).
}
\label{fig:res_R42tau}
\end{center}
\end{figure}


\end{document}